\documentclass{article}

\usepackage[T1]{fontenc}
\usepackage{lmodern}
\usepackage{amsmath,amssymb,amsthm}
\usepackage{hyperref}
\usepackage{subcaption}

\usepackage{booktabs} 

\usepackage[capitalise]{cleveref}

\newtheorem{theorem}{Theorem}
\newtheorem{lemma}{Lemma}

\newtheorem*{remark}{Remark}

\crefname{theorem}{Theorem}{Theorems}
\crefname{lemma}{Lemma}{Lemmas}
\crefname{example}{Example}{Examples}
\crefname{corollary}{Corollary}{Corollaries}

\newcommand{\listsafesectioning}{%
  \let\svsubsubsection\subsubsection
  \let\svparagraph\paragraph
  \renewcommand\subsubsection[1]{\leavevmode\svsubsubsection{##1}}%
  \renewcommand\paragraph[1]{\leavevmode\svparagraph{##1}}%
}

\usepackage{amsfonts}
\usepackage{xcolor}
\usepackage{enumitem}
\setlist{nosep}
\usepackage{tkz-graph}  
\usetikzlibrary{shapes.geometric}
\usetikzlibrary{positioning,fit}

\usepackage{xcolor}
\usepackage{caption}
\usepackage{tabularray}
\UseTblrLibrary{booktabs}
\usepackage{xspace}
\usepackage[most]{tcolorbox}
\usepackage{tabularx}

\usepackage{natbib}

\newcommand{\R}{\mathbb{R}}
\newcommand{\E}{\mathbb{E}}
\newcommand{\N}{\mathcal{N}}

\newcommand{\St}{\mathcal{S}}
\newcommand{\pr}{\mathbb{P}}

\newcommand{\I}[1]{\mathbb{I}[#1]}

\newcommand{\V}{\mathbb{V}}

\newcommand{\MinVar}{\textsc{MinVar}\xspace}
\newcommand{\StabVar}{\textsc{StableVar}\xspace}
\newcommand{\policy}{d}

\newcommand{\qt}{\quad\quad\quad}

\newcommand{\A}{\mathcal{A}}
\newcommand{\Ft}[1]{\mathcal{F}_{#1}}

\usepackage{authblk}

\renewcommand\Authfont{\fontsize{12}{14.4}\selectfont} 
\renewcommand\Affilfont{\fontsize{11}{10.8}\itshape} 

\newcommand{\Vhat}{\hat{V}_T}

\usepackage{longtable,booktabs,array}
\usepackage{calc} 

\newcommand{\titleString}{Statistical Inference for Continuous Action Bandits}

\begin{document}

  \title{\Large \titleString}
  \date{}
  \author[1,2]{Raphael C. Kim}
  \author[2]{Michele Santacatterina}
  \author[1]{Ramin Zabih}
  \author[3]{Rajarshi Mukherjee}
  \author[2]{Ivan Diaz}
  \affil[1]{Cornell Tech, Cornell University (New York, NY)}
  \affil[2]{Department of Biostatistics, Department of Population Health, New York University School of Medicine (New York, NY)}
  \affil[3]{Department of Biostatistics, Department of Biostatistics, T.H. Chan School of Public Health (Boston, MA)}
  \maketitle
  
\begin{abstract}
Statistical inference under adaptively collected data is becoming increasingly popular across e-commerce and mobile health. Many methods exist for discrete action settings, ranging from weighting to debiasing approaches. Despite the ubiquity of continuous actions in experimentation from optimal pricing to precision dosing, to the best of our knowledge, methods that support statistical inference under continuous action, adaptively collected data remain underdeveloped. 
In this work, we extend kernel-smoothed doubly robust estimation from i.i.d. data to adaptively collected data with continuous actions. We study a family of adaptively weighted estimators, establish their mean squared error rates, asymptotic normality, and characterize an estimation lower bound as a function of regret. We conclude with recommendations for experimental design informing efficiency and regret, and support our theory with an extensive simulation study.
\end{abstract}



\newpage

\section{Introduction}
The aim of this paper is to introduce methods for statistical inference in continuous action, contextual bandits. When conducting statistical inference under adaptively collected data, estimators developed for independent, identically distributed (i.i.d.) settings may exhibit non-normal behavior \citep{zhan2021offpolicy}. 
Further, the objective of regret minimization or acting optimally, must be traded off with the ability to conduct statistical inference. To achieve normality, methods range weighting to debiasing under restrictions to the exploration process; see \citep{bibaut2025demystifying}.

The discrete action case is well-understood. Let $\pi_{t,a}$ denote the propensity score for action $a \in \{ 1 \dots A \}$ at time $t$. Informally, one may achieve $\sqrt{T}$ rates and statistical inference if one constrains an \textit{exploration} lower bound $\pi_{t,a} \geq t^{-1/2}$, meaning exploitation does not happen `too fast'; such an approach is achieved through a weighting approach: \citet{bibaut2021postcontextual} utilizes an adaptive variance estimator with importance sampled empirical processes, \citet{zhan2021offpolicy} directly extends Augmented Inverse Probability Weighting (AIPW) to adaptively collected data settings \citep{zhan2021offpolicy}, and \citet{lin2025semiparametric} consider a particular estimating equation for semiparametric inference for adaptively collected data that obeys a certain generalized linear model with limited nuisance modeling. By restricting the exploitation to polynomial decay in $t$, these self-normalized approaches fail to permit minimax optimal regret. \citet{khamaru2025nearoptimal} present an alternative weighting approach which permits \textit{near-}minimax optimal regret and statistical inference, though normality occurs at theoretically slower $\log(T)$ rate for sub-optimal arms.

In contrast, the continuous action setting is not well-understood. Methods for statistical inference primarily exist for when the design is fixed, or there is one timepoint. In continuous action settings, $\pr(A_t = \pi(S_t)) = 0$, making discrete off-policy evaluation methods unusable. The approach in i.i.d. settings to resolve this is smoothing.  \citet{kennedy2017nonparametric} regress a doubly-robust (DR) pseudo-outcome on the treatment, obtaining an oracle property under a product rate condition on the nuisances. \citet{colangelo2026double} establish asymptotic normality for kernel-smoothed doubly-robust estimation of continuous-treatment effects under i.i.d. sampling with estimated propensities. \citet{chernozhukov2019semiparametric} considers semiparametric efficient policy learning with continuous actions, and \citet{mou2023kernel} considers when strict overlap fails, so the semiparametric efficiency bound may be infinite and instance-optimal risk becomes function-class dependent. 
A related work for inference under adaptively collected data is \citet{zenati2025kernel}, who considers doubly-robust scores with variance stabilization to obtain a martingale limit theorem for distributional (kernel) treatment effects under adaptively collected data. In this work, the objective is fundamentally different from us since the kernel acts on the outcome space with discrete treatments.

Several i.i.d. approaches have been extended to continuous dose, one timepoint settings for the purpose of policy learning rather than inference. \citet{kallus2018policy} extend AIPW DR estimators to continuous treatments using a kernel based approach and adapt the score to propose a variant that avoids explicit integration.
\citet{zhou2023offline} considers a tree-based DR approach that handles continuous actions with interpretable policies.
\citet{chen2016personalized} propose outcome weighted learning for personalized dose finding, advocating a trial design in which candidate dose levels are drawn from a continuous distribution over a safe range, and estimating the optimal individualized dose rule by minimizing a certain nonconvex loss; in contrast to the approaches above which learn an outcome model and optimize the plug-in, they learn the policy {directly}. \citet{zhu2020kernel} consider an extension based on kernels that properly targets the mode. \citet{zhou2018parsimonious} assume a lower dimensional representation for computational gains. 

To inform meaningful experimental design, we will need to consider tradeoffs between regret and inference. We briefly review the rich literature on minimax optimal continuous action bandit algorithms. For H\"older-$\beta$ rewards, minimax regret is $\Theta(T^{(\beta+1)/(2\beta+1)})$ \citep{kleinberg2004nearly}, attained up to polylog factors by uniform discretization into $B \asymp T^{1/(2\beta+1)}$ bins, and balancing discretization bias $TB^{-\beta}$ against finite-armed regret $\sqrt{BT}$.
`Zooming algorithms' improve upon the inefficiency introduced from uniform discretization by maintaining maintaining confidence sets over regions and exploring regions only while its confidence radius satisfies some criterion \citep{kleinberg2008multi}. \citet{krishnamurthy2020contextual} extends these approaches to contextual bandit settings.
The necessary allocation to arms are governed by the sample size for dose $a$, $T_{a} \asymp \Delta_a^{-2}\log T$ for $\Delta_a=\max_{a'} \mu(a') - \mu(a)$ with regret similarly following by a zooming-based rate rather than the uniform result. \citet{bubeck2011x} obtain an analogous result through a binary tree of coverings under a Lipschitz condition. \citet{auer2007improved} improve upon these results with a curvature assumption, to find $\tilde O(\sqrt{T})$ rates.

To the best of our knowledge, statistical inference for continuous action adaptively collected data remains a significant methodological gap. We propose a kernel-based DR approach to fill this gap. 
In contrast to related kernel-based DR approaches, we smooth the direct (outcome regression-based) estimator rather than a point-evaluated term and operate under non-batched settings with decaying propensities as \citep{kallus2018policy}, and consider the value function $\E[\mu(S, \pi(S))]$ in adaptive settings rather than the dose response curve $\E[\mu(S,a)]$ in i.i.d. settings \citep{kennedy2017nonparametric}. With respect to both of these works, we adapt a similar kernel-based DR score but handle adaptive data collection and decaying propensities.

Our contributions are as follows:
\begin{enumerate}
    \item Propose statistical inference procedures for the value function in continuous action, contextual bandits based on kernel smoothing. 
    \item Analyze mean squared error, and provide an accompanying lower bound.
    \item Characterize tradeoffs in bandwidth, exploration constraints, and statistical efficiency, to inform experimental design in continuous action settings.
\end{enumerate}

\section{Methods}\label{sec:methods}

\paragraph{Setup and Notation.}
Our setup is a continuous action contextual bandit setting. Let $S_t$ denote our state, $A_t$ denote our continuous actions with $A_t\in\A$, where $\A\subseteq\R$ is a compact interval, and $R_t$ be our reward. Data are collected adaptively, meaning $A_t \sim \pi_t(\cdot \mid \Ft{t-1}, S_t)$ for $\pi_t \in \Pi$, where $\pi_t$ is a measurable function of history $\Ft{t-1}$ and current state $S_t$, for $\Ft{t}=\sigma(\{ (S_{t'},A_{t'},R_{t'}) \}_{t'\in [t]})$ and $\Pi= \{ f \in \mathcal{H} \mid f: \mathcal{S} \times \A \rightarrow \R \}$. Throughout, we when we specify $\pi_t$, this will have history $\Ft{t-1}$ implicit. Let $\mu$ be the outcome regression of the reward conditional on state and action so $\mu(S_t,A_t) = \E[R_t \mid A_t, S_t]$. 

\paragraph{Objective.}
Our objective is to provide conditions for valid asymptotic inference on continuous action spaces, focusing on the value function from a policy $\policy: \St \rightarrow \A$
\begin{equation} V(\policy) = \E[\mu(S,\policy(S))] \end{equation}
the average reward from giving dose according to rule $\policy(S)$ under context $S$. We will require smoothing to accommodate the continuous action space, meaning we consider bias/variance tradeoffs against the smoothed value function
\begin{equation}\label{eq:Vh}
  V_h(\policy)=\E_S\Big[\int K_h\big(a-\policy(S)\big)\mu(S,a)da\Big]
\end{equation}
for $h$-Kernel $K_h(u)=h^{-1}K(\frac u h)$, where $K$ is a kernel function and $h > 0$ is the bandwidth.
\\\\
The estimator we study is defined as follows. Write $\pi_t(a\mid s)$ for the density of $A_t$ given $S_t=s$ and $\Ft{t-1}$. Let $\hat\mu_{t-1}$ be an outcome regression estimate fitted on rounds $1...t-1$, and $w_t(\policy)>0$ a weight. Both are $\Ft{t-1}$-measurable
\begin{equation}
    \Vhat(\policy) = \frac{\sum_{t \leq T} w_t(\policy) \hat{\Gamma}_t(\policy) }{\sum_{t \leq T} w_t(\policy)}
\end{equation}
where $\hat{\Gamma}_t$ represents a smoothed Doubly-Robust score given by 
\begin{align*} 
\hat{\Gamma}_t(\policy) &= \underbrace{\int K_h(a-\policy(S_t)) \hat\mu_{t-1}(S_t,a) da}_{D_t} \\
&\qt +\underbrace{\frac{K_h(A_t-\policy(S_t))}{ \pi_t(A_t\mid S_t)}\big(R_t-\hat\mu_{t-1}(S_t,A_t)\big)}_{C_t} 
\end{align*}
Here $D$ stands for direct term and $C$ stands for correction term. We consider the unweighted estimator ($w_t \equiv 1$), which requires no design knowledge, or weights that depend on the variance of the process.
\\\\
We will analyze the process in terms of martingale increments, conditional second moments, and variance:
\begin{align*}
    \Delta_t(\policy) &= \hat{\Gamma}_t(\policy) - V_h(\policy), \quad \sigma_t^2(\policy) = \V[\hat{\Gamma}_t(\policy) \mid \Ft{t-1}] \\
    \sigma_R^2(S,A) &= \V(R\mid S,A), \quad \Sigma_T(\policy)= \E[\sum_t w_t^2(\policy) \sigma_t^2(\policy)]
\end{align*}

\subsection{Theoretical Results}\label{sec:theory}
The theoretical presentation will present bias and variance of the process as a function of bandwidth, exploration rate, and weights \cref{lemma:biasVariance}. Then we will present a mean square error result \cref{thm:mse}, followed by asymptotic normality \cref{thm:ptNormality}.

We make the following regularity (\textbf{REG}), algorithm (\textbf{ALG}), and estimation assumptions (\textbf{EST}). 
\begin{enumerate}[label=\textup{(\textbf{REG\arabic*)}},leftmargin=2.6em,itemsep=1pt,nosep]
  \item Kernel regularity: $K\ge0$, symmetric, supported on $[-1,1]$, $\int K=1$, $\int uK=0$, $\mu_2(K)=\int u^2K<\infty$ (finite kernel variance), $R(K)=\int K^2<\infty$ (kernel roughness), $\kappa = \|K\|_\infty < \infty$.
  \label{ass:kernelReg}
  \item Smoothness: for all $s$, the map $a\mapsto\mu(s,a)$ is in $C^3$ on $\A$ with $\sup_{s,a}|\partial_a^2\mu(s,a)|\le C_2, \sup_{s,a}|\partial_a^3\mu(s,a)|\le C_3$. 
  \label{ass:smoothInt}
  \item Boundedness and nondegeneracy: $|R_t| < \bar R < \infty $, $|\hat\mu_{t-1}| < \bar R < \infty$, and $\underline\sigma^2 \le \sigma_R^2(s,a) \le\bar\sigma^2$ for $|a-\policy(s)|\le\delta_0$, with $\underline\sigma^2>0$.\label{ass:bdd}
\end{enumerate}
\labelcref{ass:kernelReg}-\labelcref{ass:bdd} are standard regularity conditions, imposing boundedness, differentiability, and non-degeneracy of our process. 

We impose the following algorithm assumptions:
\begin{enumerate}[label=\textup{(\textbf{ALG\arabic*})},leftmargin=2.6em,itemsep=1pt,nosep]
  \item Minimum Exploration: there is deterministic sequence of $\varepsilon_t>0$ such that for any $a \in \A$, $\bar B \geq \pi_t(a\mid s) \ge \varepsilon_t$ on $|a-\policy(s)|\le\delta_0$ for fixed $\delta_0$ and $\bar B < \infty$. \label{ass:exp}
  \item Sufficient Coverage: \label{ass:suffCoverage}
    With expectations over the history $\Ft{t-1}$, $\E[\rho_t]\to0$, where 
      \[
        \rho_t=\sup_{s,\,|v|\le h}\left|\frac{\pi_t(\policy(s)+v\mid s)^{-1}}{\E[\pi_t(\policy(s)+v\mid s)^{-1}]}-1\right|
      \]
\end{enumerate}
\labelcref{ass:exp} assumes a lower bound in exploration polynomial in $t$ as in \citet{zhan2021offpolicy,bibaut2021postcontextual}. 
\labelcref{ass:suffCoverage} adapts Condition~(8) of \citet[Thm.~1]{zhan2021offpolicy} to continuous actions, and is used to handle weights that are random through the design. Two sources of randomness are involved: the history $\Ft{t-1}$, which determines $\pi_t$ and varies from run to run, and the context, drawn from $P_S$ independently of it. The expectation in \labelcref{ass:suffCoverage} is over the history alone, across reruns of the experiment at fixed $(t,s,v)$. This condition imposes that the density near the target eventually not depend on the run. Unlike the discrete case, the supremum runs over the kernel window $|v|\le h$, and convergence is in $L^1$. Since $\rho_t$ is bounded only through $1/\varepsilon_t$ and almost sure convergence would not give the convergence in mean used in \cref{lemma:weightReg,lemma:varEst}.
For fixed $s$, let
\[
  u_t(s)=\frac{1}{\pi_t(\policy(s)\mid s)},
  \qquad
  \bar u_t(s)=\E\big[u_t(s)\big]
\]
be the inverse sampling density at the target dose in the current run and its average over runs respectively.
For $\gamma>0$, we consider weights and deterministic mean
\begin{equation}\label{eq:weightForm}
  w_t=\big(\E_S[u_t(S)]\big)^{-\gamma},
  \qquad
  \bar w_t=\big(\E_S[\bar u_t(S)]\big)^{-\gamma}
\end{equation}
where $\E_S$ integrates over contexts only. \cref{lemma:weightReg} shows that $w_t$ tracks $\bar w_t$ under
\labelcref{ass:suffCoverage}.

\begin{lemma}[Weight Regularity]\label{lemma:weightReg}
    Suppose weights $w_t$ are a power function of the inverse sampling density as in \cref{eq:weightForm}. Then, under \labelcref{ass:exp}-\labelcref{ass:suffCoverage}, we have
    \begin{enumerate}
        \item Individual Weight Concentration.
        $$ \big| \frac{w_t}{\bar w_t} - 1 \big| \leq C_g \rho_t, \quad C_g = \gamma 2^{\gamma+1} $$
        \item Summed Weight Concentration.
        Let $$ \delta_T = \E[|\frac{\sum_t w_t}{\E[\sum_t w_t]}-1|] $$
        Then, $$ \delta_T \leq 4 C_g \frac{\bar w_t \E[\rho_t]}{\sum_t \bar w_t} \rightarrow 0$$
    \end{enumerate}
    See \cref{pf:weightReg}. 
\end{lemma}
All \cref{lemma:weightReg} says is that just as the inverse densities behave like its averages, the individual weights and sum of weights do too.

Finally, we make additional estimation assumptions.
\begin{enumerate}[label=\textup{(\textbf{EST\arabic*})},leftmargin=2.6em,itemsep=1pt,nosep]
  \item Weight Control. Let $|\Delta_t | \leq M_t$ almost surely. 
  Assume $$ \max_t \frac{w_t^2 M_t^2}{\Sigma_T} = o_p(1)$$
  \label{ass:maxWeightRegularity}
  \item Undersmoothing. $$\frac{h^2 \sum_t w_t}{\Sigma^{1/2}_T} \xrightarrow{p} 0$$
  \label{ass:underSmooth}
  \item Nuisance Convergence. Our outcome regression is consistently estimated. With $\lambda_t=\|\mu-\hat\mu_{t-1}\|_\infty^2/\underline\sigma^2$, $\E[\lambda_t]\to0$.
  \label{ass:nuisance}
\end{enumerate}
\labelcref{ass:maxWeightRegularity} assumes a Lindeberg condition for continuous actions. This is a natural assumption, as under variance minimizing weights, this is equivalent to our sample size around the policy of interest growing to infinity. \labelcref{ass:underSmooth} calls for undersmoothing so the bias demonstrated in \cref{lemma:biasVariance} is negligible  \citep{kennedy2017nonparametric}. \labelcref{ass:nuisance} is the standard consistency requirement in double machine-learning literature, imposing weak restrictions on the functional form fitted \citep{chernozhukov2018double}. 

\begin{lemma}[Bias and Variance]\label{lemma:biasVariance}
Under \labelcref{ass:kernelReg}--\labelcref{ass:bdd}, \labelcref{ass:exp},
\begin{enumerate}[label=\textup{(\alph*)},leftmargin=1.6em,itemsep=2pt]
  \item \textup{Bias.} $V_h(\policy)=\E[\hat\Gamma_t\mid\Ft{t-1}]$ and $V_h(\policy)-V(\policy)=O_p(h^2)$.
  \item \textup{Variance.} Let the conditional second moment of the fitted residual be
  \[
    m_2(S_t,A_t)=\E\big[(R_t-\hat\mu_{t-1}(S_t,A_t))^2\,\big|\,\Ft{t-1},S_t,A_t\big].
  \]
  Then $\sigma_t^2=\V(\hat\Gamma_t\mid\Ft{t-1})=h^{-1}\bar V_t+r_t$, where with high probability
  \[
    \bar V_t=R(K)\,\E\Big[\frac{m_2(S,\policy(S))}{\pi_t(\policy(S)\mid S)}\Big],
    \quad |r_t|\lesssim(h\varepsilon_t)^{-1/2},
  \]
  and $c_V\le\bar V_t\le C_V/\varepsilon_t$ with $c_V=R(K)\underline\sigma^2/\bar B$, $C_V=R(K)(\bar\sigma^2+4\bar R^2)$. Here $R(K)=\int K^2$ (\labelcref{ass:kernelReg}) and $\bar B\ge\pi_t(a\mid s)$ (\labelcref{ass:exp}).
  \item \textup{Variance of the process.} Using (b),
  \begin{align*}
    \Sigma_T&\ge\frac{c_V}{h}\sum_{t\le T}\E[w_t^2]\Big(1-C_p\Big(\frac{h}{\varepsilon_t}\Big)^{1/2}\Big),\\
    \Sigma_T&\le\frac{C_V}{h}\sum_{t\le T}\frac{\E[w_t^2]}{\varepsilon_t}\Big(1+C_p(h\varepsilon_t)^{1/2}\Big),
  \end{align*}
  where $C_p=\frac{C_V}{c_V}\big(\frac{5}{4R(K)}+\frac{1}{R(K)^{1/2}}\big)$. 
\end{enumerate}
The proof is found in \cref{pf:biasVariance}.
\end{lemma}

\begin{theorem}[Mean Squared Error]\label{thm:mse}
Assume \labelcref{ass:kernelReg}--\labelcref{ass:bdd} and \labelcref{ass:exp}--\labelcref{ass:suffCoverage}. Let $\delta_T$ be as in \cref{lemma:weightReg}, and let $ \kappa_T=\frac{\E\sum_tw_t}{\sum_t\varepsilon_t^{\gamma}} $. 
Suppose
\begin{enumerate}[label=\textup{(\roman*)},leftmargin=2.2em,itemsep=1pt,nosep]
  \item $\kappa_T^4\,\delta_T\to0$;
  \item $\sup_T\E\big[\big(\sum_t w_t^2\sigma_t^2/\Sigma_T\big)^2\big]<\infty$.
\end{enumerate}
Then
\begin{align*}
  \mathrm{MSE}\big(\hat V_{T}\big) &= \underbrace{\frac{\Sigma_T}{\big(\E[\sum_t w_t]\big)^{2}}}_{\text{variance}}
    +\underbrace{\big(V_h(\policy)-V(\policy)\big)^{2}}_{\text{bias}^2}
\end{align*}
For deterministic weights, $\delta_T=0$: conditions (i)--(ii) are not needed and the identity holds exactly. The proof is found in \cref{pf:mse}.
\end{theorem}

\paragraph{Asymptotic Normality}
Define the variance estimator as
\begin{equation}\label{eq:varEst}
    \hat{\Sigma}_T(\policy) = \sum_t w_t^2(\policy) (\hat{\Gamma}_t(\policy)-\hat{V}_T(\policy))^2 
\end{equation}
Then we have the following results.
\begin{lemma}[Variance Consistency]\label{lemma:varEst}
    Assume \labelcref{ass:kernelReg}--\labelcref{ass:bdd}, \labelcref{ass:exp}-\labelcref{ass:suffCoverage}, \labelcref{ass:maxWeightRegularity}, and \labelcref{ass:nuisance}. 

    We make two claims. First, that the oracle variance has
    $$ \left| \frac{\sum_t w_t^2 \sigma_t^2}{\Sigma_T} - 1 \right| =o_P(1)$$
    We then claim that our variance estimator defined in \cref{eq:varEst} is consistent.
    $$ \left| \frac{\hat{\Sigma}_T}{\Sigma_T} -1 \right| = o_P(1)$$
    The proof is found in \cref{pf:varEst}.
\end{lemma}

We are now fit to introduce a result on inference.
\begin{theorem}[Asymptotic Normality]\label{thm:ptNormality}
Assume the conditions of \cref{lemma:varEst}, 
\begin{align*}
  \hat\Sigma_T(\policy)&=\sum_{t\le T}w_t^2(\policy)\big(\hat\Gamma_t(\policy)-\Vhat(\policy)\big)^2 \\
  b_h(\policy)&=V_h(\policy)-V(\policy)
\end{align*}
Then
\[
  \frac{\Vhat(\policy)-V(\policy)-b_h(\policy)}{\hat\Sigma_T(\policy)^{1/2}/\sum_tw_t(\policy)}
  \rightsquigarrow\N(0,1)
\]
If in addition \labelcref{ass:underSmooth} holds, the same limit holds with $b_h(\policy)$ omitted. The proof is found in \cref{pf:ptNormality}.
\end{theorem}

\begin{remark}[The bias term]
By \cref{lemma:biasVariance}(a),
\[
  b_h(\policy)=\frac{h^2}{2}\mu_2(K)\E_S\big[\partial_a^2\mu(S,\policy(S))\big]+O(h^3)
\]
the curvature multiplied by the spread of the kernel $\mu_2 = \int u^2 K$. Since $b_h$ is unknown, an interval for $V(\policy)$ itself requires undersmoothing, as in the second statement, or bias correction.
\end{remark}

We now give an accompanying lower bound. Let $\mathcal P(L,\nu)$ be the environments with a fixed context law $P_S$, rewards $R=\mu(S,A)+\epsilon$, $\epsilon\mid S,A\sim\N(0,\nu^2)$, and $|\partial_a^2\mu|\le L$ almost surely. For $P_0$, let $\underline\Delta_\policy=\inf_s[\max_a\mu_0(s,a)-\mu_0(s,\policy(s))]$ be the gap of the target and $\mathrm{Regret}_T$ the cumulative regret under $P_0$.
\begin{theorem}[Lower Bound]\label{thm:lb}
Fix an algorithm and $P_0\in\mathcal P(L/2,\nu)$, and let
\[
  G=\sum_t\E\Big[\sup_{|a-\policy(S)|\le\bar h}\pi_t(a\mid S)\Big]
\]
be the expected concentration around the target. Then,
\begin{enumerate}[label=\textup{(\alph*)},leftmargin=1.6em,itemsep=2pt]
  \item \textup{Estimation.} For $\E_0G$ large enough,
  \[
    \inf_{\hat V}\sup_{P\in\mathcal P(L,\nu)}\E_P\big[(\hat V-V_P(\policy))^2\big]\gtrsim\Big(\frac{\nu^2}{\E_0G}\Big)^{4/5}
  \]
  \item \textup{Regret.} If $\mu_0(s,\cdot)$ is Lipschitz, then for any $\hat V$ with $\sup_{P\in\mathcal P(L,\nu)}\E_P[(\hat V-V_P(\policy))^2]$ small enough,
  \begin{align*}
    \E_0[\mathrm{Regret}_T]&\sup_{P\in\mathcal P(L,\nu)}\E_P\big[(\hat V-V_P(\policy))^2\big] \gtrsim\underline\Delta_\policy\nu^2 
  \end{align*}
\end{enumerate}
Constants and conditions are given in \cref{pf:lb}.
\end{theorem}
Estimation is limited by how much the algorithm concentrates near the target (a), and that costs regret times worst-case MSE at least of order $\underline\Delta_\policy\nu^2$ (b), so halving the error requires four times the regret.
\subsection{Interpretation and Guidelines for Experimental Design}\label{sec:interp}

In this section, we describe how to interpret the results and design experiments in continuous action settings.
Following the minimum exploration conditions of \citet{zhan2021offpolicy,bibaut2021postcontextual}, we propose an adaptation of $\epsilon-$greedy to continuous dose experiments as follows. Let the collection policy be given by
\[
  \pi_t=(1-\rho_t)\tilde\pi_t+\rho_tU
\]
where $\tilde\pi_t$ is the algorithm's own choice (often regret-minimizing), $\rho_t=t^{-\alpha}$ is a measure of forced exploration, and $U$ is a fixed exploration density with $\underline U=\inf_aU(a)$. The exploration constraint gives $\pi_t(a\mid s)\ge\varepsilon_t=\rho_t\underline U$, so every target's kernel window receives at least $2hE$ observations in expectation, 
\[
  E=\sum_{t\le T}\varepsilon_t\asymp
  \begin{cases}
    T^{1-\alpha}/(1-\alpha) & \alpha<1\\
    \log T & \alpha=1
  \end{cases}
\]
$E$ is referred to as exploration mass, and, consistency requires $hE\to\infty$, which is possible for $\alpha<1$, possible at $\alpha=1$ only if $h\gg1/\log T$, and impossible for $\alpha>1$.

Under this setup, we discuss the notion of on- and off-policy targets, exploitation, weighting, and design.

\paragraph{On- and off-target policies.}
We first discuss the policy choice, $\policy$. We characterize this as on- and off-target. By \cref{lemma:biasVariance}(b), the variance of a round's score is governed by
\[
  \Psi_t(\policy)=\E_S\int\frac{K_h(a-\policy(S))^2}{\pi_t(a\mid S)}da
  \le\frac{R(K)}{h\big(\varepsilon_t+(1-\rho_t)q_t\big)}
\]
where $q_t=\inf_s\inf_{|a-\policy(s)|\le h}\tilde\pi_t(a\mid s)$ is the density the algorithm places near $\policy$, the target of interest.
A target is \emph{off-target} when $q_t \approx 0$, that is, the target is reached primarily through exploration so $\Psi_t\asymp(h\varepsilon_t)^{-1}$, and the rate of $\Psi_t$ governed by exploration mass $E$.
A target is is \emph{on-target} when $q_t\to q>0$, that is, $\Psi_t\asymp R(K)/(hq)$ is essentially constant and the target may enjoy an i.i.d. rate in $T$. 

\paragraph{Exploitation in Continuous Dose Settings.}

A natural thought for exploiting means playing the greedy dose exactly. However, for inferential purposes, this is not desirable as the kernel estimator for a policy near an exploitation target will fail to account for \textit{exact} exploited dose. This means that exploit rounds only the outcome model, limiting efficiency gains. We recommend jittering the exploited dose, playing $A_t\sim\mathrm{Unif}[g_t(S_t)-b,g_t(S_t)+b]$ for greedy dose $g_t(S_t)$ played with bandwidth $b$ so that exploit rounds land with a density. 
Empirically, we show that $b>0$ improves inference; theoretically the tradeoff is as follows: regret increases by about $b^2$ per round, but as $b\geq h$, the bandwidth choice, exploit rounds count toward estimating the target. A width of $b\approx 2h$ empirically safely buys a large gain in effective sample size at negligible regret cost.

\paragraph{Weighting Choice.}
We consider three $\Ft{t-1}$-measurable weightings: \MinVar, $w_t=\Psi_t(\policy)^{-1}$, \StabVar, $w_t=\Psi_t(\policy)^{-1/2}$, and $w_t\equiv1$. The first two adapt to the target through $\Psi_t$. Of the two propensity-based weights, \MinVar has the weaker central limit condition, $hE\to\infty$, and remains valid at $\alpha=1$ where \StabVar does not. However, \StabVar, which depends on $\Psi_t$ only through its square root, is less sensitive when the density near the target varies with the history (or higher instability in the algorithm). The unweighted estimator needs no knowledge of the design, but off-target its variance is $1/(1-\alpha^2)$ times that of \MinVar (\cref{tab:weighting-efficiency}) while on-target all three attain the same rate.

\begin{table}[ht]
\centering
\setlength{\abovecaptionskip}{\baselineskip}
\setlength{\belowcaptionskip}{\baselineskip}
\caption{Off-target weighting ($\kappa_T\asymp1$) under $\rho_t=t^{-\alpha}$, $\alpha<1$: efficiency relative to \MinVar, and the form taken by (EST1). The conditions coincide in order, $hT^{1-\alpha}\to\infty$, for polynomial schedules.}
\label{tab:weighting-efficiency}
\renewcommand{\arraystretch}{1.6}
\setlength{\tabcolsep}{4pt}
\begin{tabular}{@{}llll@{}}
\toprule
Method & $w_t$ & $\mathrm{RelEff}$ & \labelcref{ass:maxWeightRegularity} \\
\midrule
Unweighted & $1$ & $\dfrac1{1-\alpha^2}$ & $h\varepsilon_T^2\sum_t\varepsilon_t^{-1}\to\infty$ \\
\StabVar & $\Psi_t^{-1/2}$ & $\dfrac{(1-\alpha/2)^2}{1-\alpha}$ & $hT\varepsilon_T\to\infty$ \\
\MinVar & $\Psi_t^{-1}$ & $1$ & $hE\to\infty$ \\
\bottomrule
\end{tabular}
\end{table}

\paragraph{Weight concentration and exploited targets.}
With random weights, \cref{thm:mse} needs $\kappa_T^4\delta_T\to0$. Off-target $\kappa_T\asymp1$, so this is $\delta_T\to0$ (\cref{lemma:weightReg}). At an exploited target $\kappa_T\asymp T^{\alpha\gamma}$, so the design must stabilize at rate $\delta_T=o(T^{-4\alpha\gamma})$. If $\delta_T$ decays no faster than $T^{-1}$, this fails once $4\alpha\gamma\ge1$: $\alpha\ge\frac14$ for \MinVar\ and $\alpha\ge\frac12$ for \StabVar. 
This affects the MSE only, not normality; fixed weights such as $w_t=\varepsilon_t^\gamma$ give $\delta_T=0$, at some cost in efficiency.

\paragraph{Bandwidth.}
Off-target, \cref{thm:mse} gives
\[
  \mathrm{MSE}\big(\hat V_T\big)\asymp\frac{1}{hE}+h^4,
  \quad h^\star\asymp E^{-1/5},
  \quad \mathrm{MSE}^\star\asymp E^{-4/5}
\]
while inference requires $E^{-1}\ll h\ll E^{-1/5}$, so $h^\star$ lies on the boundary and intervals require undersmoothing. With $\alpha=0$, $E\asymp T$: the i.i.d. results hold with $T$ replaced by $E$, or by $qT$ on-target.

\subsection*{Recommendations}
\begin{itemize}[leftmargin=1.5em,itemsep=2pt]
  \item \emph{Design for the goal.} For a dose--response curve or many doses, we require minimum exploration: pick a $U$ that covers the action space and small $\alpha$. For the policy the algorithm favors, exploit with $b \approx 2h$ around the greedy dose.
  \item \emph{Exploration.} Keep $\alpha<1$; a smaller $\alpha$ improves the rate $T^{-4(1-\alpha)/5}$ at the cost of regret.
  \item \emph{Weights.} Use \MinVar when the density near the target is stable or $\alpha$ is near $1$, \StabVar when it varies with the history, and $w_t\equiv1$ when the design is unknown.
  \item \emph{Bandwidth.} Take $h\asymp E^{-1/5}$ for estimation and undersmooth, e.g. $h\asymp E^{-1/5}/\log E$. Note, while undersmoothing \labelcref{ass:underSmooth} offers a route for normality, we practically recommend a bias correction approach. As noted in \citet{kennedy2017nonparametric}, ``it is mostly a technical device''.  
  \item \emph{Pre-specify and log.} Fix $h$, $b$, $\rho_t$ and $U$ before collection and log the propensity function for flexibility in off-policy evaluation.
\end{itemize}
\section{Simulations}\label{sec:sims}
In this section, we carry out an extensive simulation study to verify our theoretical findings.

\paragraph{Setup}
We carry out 1000 simulations of the process below, computing 95\% coverage and mean squared error averaged over the simulations. More formally, for an estimate $\hat V^{(n)}$ and an interval $[\hat V^{(rn}\pm q^{(n)}\widehat{\mathrm{se}}^{(n)}]$, and the true value $V(\policy)$ is computed by quadrature. We report
$\mathrm{Coverage} = \frac1N\sum_{n=1}^N\I{\big|\hat V^{(n)}-V(\policy)\big|\le q^{(n)}\widehat{\mathrm{se}}^{(n)}} $ and $ \mathrm{MSE} = \frac1N\sum_{n=1}^N\big(\hat V^{(n)}-V(\policy)\big)^2$, for $q^{(r)}$ the $t_{0.975}$ quantile of each method (\cref{app:sim:base}).
\\\\
Each round $t=1,\dots,T$ ($T=5000$): state $S_t\sim\mathrm{Unif}[-1,1]$, dose $A_t\in[0,1]$ from the logger, reward $R_t=\mu(S_t,A_t)+\epsilon_t$ for  $\epsilon_t \sim N(0,0.5^2)$.

\textbf{Outcome Generation.}
We cover three cases:
\begin{description}[leftmargin=1.2em,nosep]
\item[Quadratic.] The first case is a quadratic outcome mechanism, in which the model will be correctly specified. The purpose of this is to isolate effects from adaptivity. The outcome model is given by $\mu(s,a)=2a(1-a)+0.5 sa$, yielding an optimal dose of $g^*(s)=\tfrac12+\tfrac s8$, and $\partial_a^2\mu\equiv-4$.
\item[Bimodal.] The second case is a bimodal outcome model.
Two Gaussian ridges of width $0.09$ at $a=0.25$ and $0.75$ are simulated with amplitudes $0.85\mp 0.20 s$ and a crossing at $s=0$: $\mu(s,a) =(0.85-0.2s) e^{-(a-0.25)^2/(2\cdot0.09^2)} +(0.85+0.2s) e^{-(a-0.75)^2/(2\cdot0.09^2)}$.
\item[RFF.] Our final case is a challenging nonparametric function.
A random-Fourier-feature draw ($64$ features) approximating a squared-exponential Gaussian process is used, with length-scales $(\ell_s,\ell_a)=(0.6,0.12)$, rescaled to mean $0.3$, sd $0.35$ on a grid.
\end{description}

\textbf{Adaptive logger}
Exploration probability $\rho_t=t^{-\alpha}$ with $\alpha\in\{0,0.25,0.5,0.75,0.9\}$.
With $U$ uniform on $[0,1]$, the exploration floor of \cref{sec:interp} is $\varepsilon_t=\rho_t$.
\begin{itemize}[leftmargin=1.2em]
\item \textbf{Outcome model / greedy.} Least-squares polynomial in $s^ja^k$, $j\le2$, $k\le d$, with dose degree $d\in\{2,\dots,8\}$ using (forward in time) 5-fold cross-validation \citep{bibaut2021postcontextual}.
Greedy dose $\hat g_t(s)=\arg\max_{a\in\{0.05,\dots,0.95\}}\hat\mu_t(s,a)$.
\item \textbf{Exploit step.} We test both bandwidth exploitation ($b$) and point dosing to support our recommendations.
\begin{itemize}
\item Bandwidth Exploitation: $A_t\sim\mathrm{Unif}\,B_t(S_t)$ with $B_t(s)=[\hat g_t(s)-b,\hat g_t(s)+b]$, $b=0.08$.
\item Point Dose: $A_t=\hat g_t(S_t)$ exactly.
Exploit rounds then carry zero weight in the kernel correction, so only exploration rounds inform it.
\end{itemize}
\item \textbf{Propensities} are the known density $\pi_t(a\mid s)=\rho_t+(1-\rho_t) \I{a\in B_t(s)}/|B_t(s)|$, where $|B_t(s)|$ is the length of the clipped band.
\end{itemize}

\textbf{Target policies}
\begin{center}\small
\begin{tabularx}{\columnwidth}{@{}ll>{\raggedright\arraybackslash}X@{}}\toprule
label & $\policy(s)$ & role\\\midrule
\texttt{on-target} & $g^*(s)$ & true optimum \\
\texttt{off-target} & $o+0.1s$ & exploration data only\\
\bottomrule
\end{tabularx}\end{center}
The \texttt{off-target} offset is $o=0.15$, $0.50$, $0.25$ in the quadratic, bimodal and RFF respectively, away from exploitation of the algorithm.

\begin{description}[leftmargin=1.2em,nosep]
    \item[Proposed.] Kernel DR scores with predictable \MinVar weights and robust bias correction (\cref{sec:methods}).
    \item[Direct plug-in.] $\hat V=\frac1n\sum_t\hat\mu(S_t,\policy(S_t))$ with the shared outcome fit $\hat\mu$.
    \item[Kennedy cross-fit.] 5-fold CV and cross-fitting, treating the data as i.i.d.\ \citep{kennedy2017nonparametric}.
    \item[Kallus--Zhou.] Self-normalized kernel IPW with their plug-in MSE-optimal bandwidth.
    \item[Zhan.] Binned AIPW over $B$ dose bins ($B\in\{4,\dots,28\}$ by forward CV) with their variance-stabilizing weights $(\E_S[1/e_t(S)])^{-1/2}$, where $e_t(s)$ is the probability of the target's bin at $s$.
\end{description}
All intervals are Wald-type, differing only in the variance estimator with i.i.d. / non-i.i.d. procedures for hyperparameter selection. Further details are in \cref{app:sim:base}.

\subsection{Results.}

Our primary results are shown in \cref{fig:coverage} and \cref{fig:mse}, under bandwidth exploitation ($b > 0$).
\begin{figure}[ht!]
  \centering
  \includegraphics[width=0.85\textwidth]{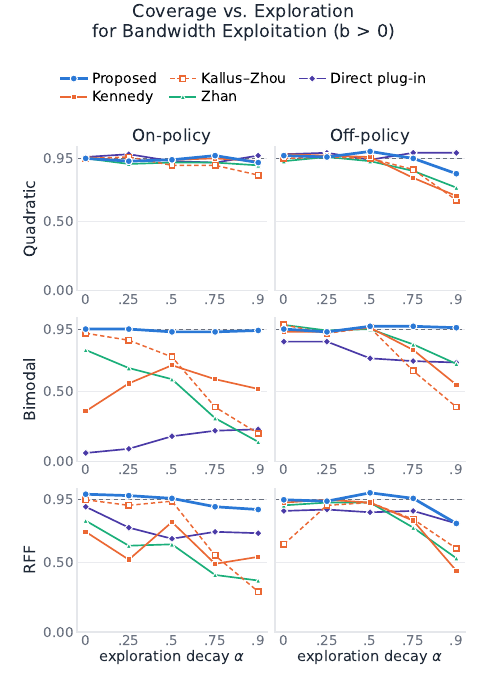} 
  \caption{Coverage for 95\% confidence intervals for the policy value $V$ as exploration decays ($\rho_t = t^{-\alpha}$).
  Columns: target policy at the true optimum (on-target), and a fixed off-optimum rule (off-target).
  Rows: outcome mechanisms (Quadratic, Bimodal, RFF).
  The logger exploits with half-width $b = 0.08$ around its greedy dose; the kernel bandwidth $h$ is chosen along the shared forward cross-validated path.
  $T = 5000$ rounds, 1000 replications per cell.
  Dashed line: nominal level $0.95$.}
  \label{fig:coverage}
\end{figure}

\begin{figure}[ht!]
  \centering
  \includegraphics[width=0.85\textwidth]{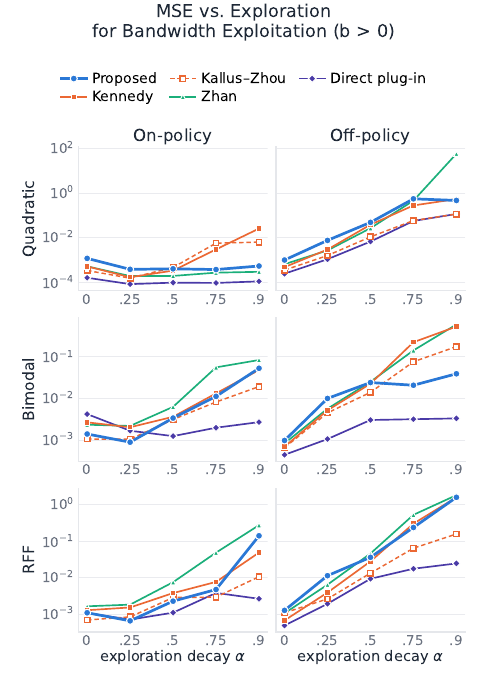}
  \caption{Mean squared error of the policy-value estimates (log scale), same design as Figure~\ref{fig:coverage}.}
  \label{fig:mse}
\end{figure}

\cref{fig:coverage} shows that our method (blue circles) attains coverage at or near the nominal level for all $\alpha \le 0.75$, for both the on-target and off-target policies and across the quadratic, bimodal, and RFF outcome models.
In contrast, the baselines undercover in the misspecified settings of bimodal and RFF.
Interestingly, \cref{fig:mse} shows our mean squared error is comparable to that of the other kernel-based estimators.
The direct plug-in attains the lowest mean squared error.
At $b=0$, point dose exploitation, we see lower coverage as shown in \cref{fig:atomCmp}.

We consider coverage in more detail.
Under the quadratic model the outcome model is correctly specified and bias is small, so all methods cover well for $\alpha \le 0.5$.
Interestingly, off-target, at high exploitation, we cover better than competing methods except for the direct method, which is perfectly specified.
This case offers a nice sanity-check to verify the effect of adaptivity, as we observe lower coverage under higher exploitation and off-target policies.

The bimodal and RFF settings demonstrate the inability of competing baselines to perform inference in this setting.
In both bimodal and RFF, the on-target (optimal) policy falls at a maximum of the outcome function where bias is largest leading to undercoverage.
Most baselines even undercover under uniform logging ($\alpha = 0$), worsening as exploration decays: for example, at $\alpha = 0.75$, coverage is $0.59$ for \citet{kennedy2017nonparametric}, $0.39$ for \citet{kallus2018policy}, $0.31$ for \citet{zhan2021offpolicy}, and $0.22$ for the direct plug-in.
In contrast, off-target policies lie in flatter regions so competing methods cover until exploration is scarce.

\cref{fig:mse} demonstrates our mean squared error is comparable to that of the other kernel-based estimators.
For on-target policies, all kernel-based methods attain similar error.
For off-target policies, our findings are driven by the bandwidth $h$ since smoothing bias grows as $\mu''h^2$ while variance falls as $1/(Th)$.
In the quadratic and RFF models, the outcome model is nearly flat, so a wide bandwidth adds little bias while reducing variance.
However, our narrower bandwidth, plus variance from bias correction yields higher relative error.
On the bimodal model, there is greater curvature so our narrower bandwidth pays off.

Finally, \cref{fig:teff} plots coverage against the effective sample size $T_{\mathrm{eff}}=\sum_t\pr(|A_t-\policy(S_t)|\le h)$, roughly the number of logged observations near the target dose.
We include this plot to demonstrate that our method covers at the nominal level once $T_{\mathrm{eff}}$ reaches a certain level, while competing baselines fail in comparison.
For on-target policies, competing baselines coverage worsens as $T_{\mathrm{eff}}$ grows, because their intervals shrink without correction.
At off-target policies, their coverage is worst when $T_{\mathrm{eff}}$ is small, because their variance estimates ignore the adaptive design.

\begin{figure}[ht!]
  \centering
  \includegraphics[width=0.85\textwidth]{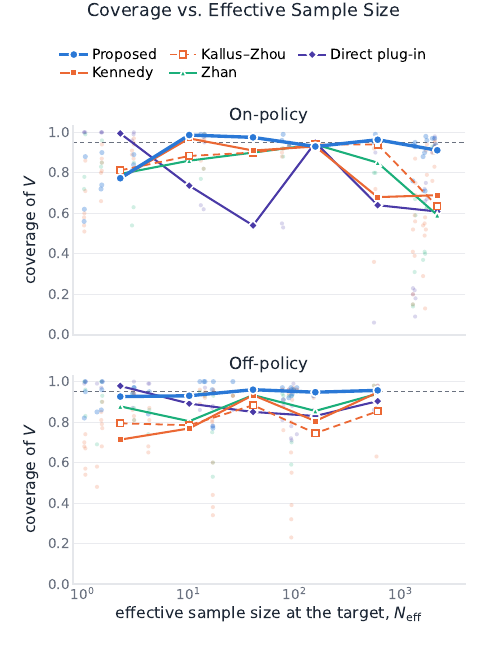}
  \caption{Coverage of nominal 95\% intervals versus the effective sample size $T_{\mathrm{eff}}$ at the target, for on-target (top) and off-target (bottom) policies.
  Faint points are individual configurations (outcome model, exploration decay $\alpha$, and logger); lines average coverage within bins of $T_{\mathrm{eff}}$.
  Dashed line: $0.95$.}
  \label{fig:teff}
\end{figure}

\section{Conclusion}
We propose a methodology for performing statistical inference on contextual bandits under continuous actions. The setting of continuous action bandits is challenging due to smoothing requirements and non-trivial exploration bounds, which affect the whole pipeline from the deployed algorithm to the analysis choices across bandwidth and weighting. We characterize bandwidth and weighting tradeoffs needed to obtain optimal estimation rates and inference, recovering i.i.d. results up to exploration constraints. We conclude with design implications, recommending the practitioner to choose \MinVar choice with known propensities and exploitation settings, \StabVar when instability arises, and no weighting with little information on the log, supporting our theoretical findings with an extensive simulation study. 

\newpage 

\renewcommand{\refname}{}

\clearpage 

\appendix


\renewcommand\Authfont{\fontsize{12}{14.4}\selectfont} 
\renewcommand\Affilfont{\fontsize{11}{10.8}\itshape} 

\crefname{section}{Appendix}{Appendices}

\counterwithin{equation}{section}
\counterwithin{figure}{section}
\counterwithin{table}{section}
\pagenumbering{arabic}
\setcounter{page}{1}

\renewcommand{\thesection}{Appendix~\Alph{section}}
\appendix 

\begin{center}
    \textbf{\large Supplementary Materials to\\ ``\titleString''}\\ \vspace{0.25cm}
\end{center}

\section{Proofs}

\subsection{\texorpdfstring{Proof of \cref{lemma:weightReg}}{}}
\label{pf:weightReg}
\begin{proof}
\listsafesectioning

\paragraph{(i) Individual Weight Control.}
Recall, we define
$$ u_t(s) = \frac{1}{\pi_t(\policy(s) \mid s)}, \bar u_t = \E u_t$$
and $g(x)=x^{-\gamma}, \gamma > 0$. Then rewrite $\frac{w_t}{\bar w_t}$, for $\bar w_t = g(\E[u_t(S)])$, as
\begin{align*}
    \frac{w_t}{\bar w_t}= (\frac{\E u}{\E \bar u})^{-\gamma}
\end{align*}

We claim that $|\frac{\E u}{\E \bar u}-1 | \leq \rho_t $, proven below. With this in hand, we can take sufficiently large $t$ s.t. $\rho_t \leq 1/2$ and invoke the mean value theorem to conclude
\begin{align*}
    |\frac{w_t}{\bar w_t}-1| \leq |\frac{\E u}{\E \bar u}-1| \sup_{|x-1| \leq 1/2} \gamma x^{-\gamma-1} \leq \rho_t \gamma (1/2)^{-\gamma-1}
\end{align*}
Then, 
$$ w_t = \bar w_t (1+ \varepsilon_t^w), \varepsilon_t^w=w_t/\bar w_t - 1$$

\textbf{Proof of Claim}\\
Rearrange the inequality in \labelcref{ass:suffCoverage},
\begin{align*}
    \bar u (1-\rho) \leq &u \leq \bar u (1 + \rho)\\
    \iff \E \bar u (1-\rho) \leq & \E u \leq \E \bar u (1 + \rho)\\
    \iff |\frac{\E u}{\E \bar u}-1 | &\leq \rho_t 
\end{align*}

\paragraph{(ii) Summed Weight Control.}
From (i), $ \sum_t w_t = \sum_t \bar w_t (1+\varepsilon_t^w)$ and $\E[\sum_t w_t] = \sum_t \bar w_t (1+\E \varepsilon_t^w)$. 
Thereafter, taking the ratio, subtracting 1, and renormalizing, we have that
\begin{align*}
    \frac{\sum_t w_t }{\E[\sum_t w_t ]} -1 &= \frac{\sum_t w_t (1 + \varepsilon_t^w)-\sum_t \bar w_t (1 + \E[\varepsilon_t^w])}{\sum_t \bar w_t (1 + \E[\varepsilon_t^w])}\\
    &= \frac{\sum_t \bar w_t (\varepsilon_t^w-\E \varepsilon_t^w)}{\sum_t \bar w_t ( 1+ \E[\varepsilon_t^w])}
\end{align*}
Note, for large $T$, the denominator is at least $1/2 \sum_t \bar w_t$ since $\E[\varepsilon_t^w]$ can be taken small, $\leq 1/2$. 
\\\\
Take absolute values, and by triangle inequality, we ahve
\begin{align}
    \delta_T \lesssim \frac{\sum_t \bar w_t \E[\rho_t]}{\sum_t \bar w_t} 
\end{align}
which by Toeplitz's Lemma, goes to 0.

\end{proof}

\subsection{\texorpdfstring{Proof of \cref{lemma:biasVariance}}{}}
\label{pf:biasVariance}
\begin{proof}
  \listsafesectioning

  Recall the score and the definitions of the direct term $D_t$ and correction term $C_t$:
    $$ \hat{\Gamma}_t = \underbrace{\int K_h(a-\policy(S_t)) \hat \mu_{t-1}(S_t,a) da}_{D_t} +\underbrace{\frac{K_h(A_t-\policy(S_t))}{\pi_t(A_t\mid S_t)}\big(R_t-\hat\mu_{t-1}(S_t,A_t)\big)}_{C_t} $$
    
    \subsubsection{Bias.} We claim $V_h(\policy)=\E[\hat\Gamma_t\mid\Ft{t-1}]$ and $V_h(\policy) - V(\policy) = O_p(h^2)$. The calculations will bound $\E[\hat{\Gamma}_t \mid \Ft{t-1}, S_t=s]$, then marginalize to conclude the result.

    \paragraph{1. Simplification.}

    We first compute the conditional mean at a fixed context. The propensity cancels against the sampling density since
    \begin{align*}
      \E\big[C_t\mid S_t=s,\Ft{t-1}\big] &=\int\pi_t(a\mid s)\frac{K_h(a-\policy(s))}{\pi_t(a\mid s)}\big(\mu-\hat\mu_{t-1}\big)(s,a)da \\
      &=\int K_h\big(a-\policy(s)\big)\big(\mu(s,a)-\hat\mu_{t-1}(s,a)\big)da
    \end{align*}
     Adding the direct term, the $\hat\mu_{t-1}$ terms cancel, yielding
     \[
       \E\big[\hat\Gamma_t\mid S_t=s,\Ft{t-1}\big]=\int K_h\big(a-\policy(s)\big)\mu(s,a)da
     \]
    Since $S_t$ is independent of $\Ft{t-1}$, marginalizing gives
    \[
      \E\big[\hat\Gamma_t\mid\Ft{t-1}\big]=V_h(\policy)
    \]
    almost surely.

    \paragraph{2. Taylor Expansion.}
    In order to bound this expression, we will take a Taylor Expansion.
    \begin{align*}
        \mu(s,\policy(s)+hu) &=\mu(s,\policy(s)) + hu \partial_a \mu(s,\policy(s)) \\
        & \qt + \tfrac12h^2 u^2 \partial_a^2 \mu(s,\policy(s)) + \tfrac 1 6 h^3 u^3 \partial_a^3 \mu(s, \policy(s)+\theta hu)
    \end{align*}
    for some $\theta\in(0,1)$. 
    Integrating against $K$ with $\int K=1$, $\int uK=0$, $\int u^2 K = \mu_2(K)$,    
    $$\E[\hat\Gamma_t\mid\Ft{t-1}, s]=\mu(s,\policy(s))+\tfrac12 h^2\mu_2(K)\,\partial_a^2\mu(s,\policy(s))+ R_3$$
    where $$ | R_3 | \leq \tfrac 1 6 h^3 C_3 \int |u|^3 K(u) du \leq \tfrac 1 6 h^3 C_3 $$

    Taking $\E_S$ and subtracting $V(\policy)=\E_S[\mu(S,\policy(S))]$, 
    $$ V_h(\policy)-V(\policy)=\tfrac12h^2\mu_2(K)\E_S[\partial_a^2\mu(S,\policy(S))]+o_p(h^2)=O_p(h^2) $$

    \subsubsection{Variance.}

    Recall definitions $\sigma_t^2 = \V(\hat\Gamma_t\mid\Ft{t-1})$, $m_2(s,a)=\sigma_R^2(s,a)+(\mu-\hat\mu_{t-1})^2(s,a)$, so $\underline\sigma^2\le \bar m_2=m_2\le\bar\sigma^2+4\bar R^2 $, $R(K)=\int K^2$, and $\bar R$ is an upper bound on reward. 
    We claim the per-round variance is given by the primary variance term and remainder
    $$\sigma_t^2 = h^{-1} \bar V_t + r_t$$
    for
    \begin{align*}
      \bar V_t = R(K) \E_S\left[\frac{m_2(S,\policy(S))}{\pi_t(\policy(S)\mid S)} \right],
      \qquad |r_t| \lesssim (h \varepsilon_t)^{-1/2},
    \end{align*}
    with high probability.
    We additionally claim this quantity is upper and lower bounded as
    $$     c_V \le \bar V_t \le C_V/\varepsilon_t  $$
    for $c_V=R(K)\underline\sigma^{2}/\bar B$ and $C_V=R(K)(\bar\sigma^{2}+4\bar R^{2})$.\\\\
    We now show the first claim, considering $\hat\Gamma_t=D_t+C_t$ and computing variance term by term on $\V(D_t)+\V(C_t) + 2 Cov(D_t, C_t)$.

    \paragraph{Correction term, second moment.} One factor of $\pi_t$ cancels leaving
    \begin{align*}
      \E[C_t^2\mid S_t=s,\Ft{t-1}] &=\int \pi_t(a\mid s)\frac{K_h(a-\policy(s))^2}{\policy_t(a\mid s)^2}m_2(s,a)da \\
      &=\int\frac{K_h(a-\policy(s))^2}{\pi_t(a\mid s)}m_2(s,a) da \\
      &= h^{-1} \int_{-1}^{1} \frac{K(u)^2 m_2(s, \policy + hu)}{\pi_t(\policy + hu \mid s)}  du
    \end{align*}
    by substitution $a = \policy + hu$ which implies the kernel is given by $K_h^2 da = h^{-2}K^2 h du$. 

    On the kernel's support $|u|\le1$ we have $|a-\policy(s)|\le h\le\delta_0$, so by \labelcref{ass:bdd} $\underline\sigma^2\le m_2(s,a)\le\bar\sigma^2+4\bar R^2$, and by \labelcref{ass:exp} $1/\bar B\le1/\pi_t(a\mid s)\le1/\varepsilon_t$. Bounding the integrand below and above, recalling $\int_{-1}^{1}K(u)^2du=R(K)$,
    \begin{align*}
      \E[C_t^2\mid S_t=s,\Ft{t-1}] &\geq h^{-1}\int_{-1}^{1}K(u)^2\frac{\underline\sigma^2}{\bar B}du =\frac{R(K)\underline\sigma^2}{h\bar B} \\
      \E[C_t^2\mid S_t=s,\Ft{t-1}] &\leq h^{-1}\int_{-1}^{1}K(u)^2\frac{\bar\sigma^2+4\bar R^2}{\varepsilon_t}du =\frac{R(K)\big(\bar\sigma^2+4\bar R^2\big)}{h\varepsilon_t}
    \end{align*}
    In other words, for $c_V, C_V$ defined as in the result,
    \[
      \frac{c_V}{h}\leq \E[C_t^2\mid S_t=s,\Ft{t-1}]\le\frac{C_V}{h\varepsilon_t}
    \]

    \paragraph{Correction term, first moment.}
    By the cancellation in the bias step, 
    $$ | \E[C_t\mid\Ft{t-1}] | \le \E_S \int K(u)|\mu-\hat\mu_{t-1}|du \le2\bar R $$
    
    Hence, 
    \begin{align*}
    \V(C_t\mid\Ft{t-1}) &=\E[C_t^2\mid\Ft{t-1}]-(\E[C_t\mid\Ft{t-1}])^2=h^{-1}\bar V_t+O_P(1)
    \end{align*}

    \paragraph{Direct term.} Since $\int K_h = 1$ and $|\hat \mu_{t-1}| \leq \bar R$, $|D_t|\le\bar R$. As a result, $\V(D_t \mid \Ft{t-1}) \leq \bar R^2 = O(1)$.

    \paragraph{Cross term.} By Cauchy--Schwarz and the upper bound on $\bar V_t$,
    \begin{align*} 
        | Cov(D_t, C_t \mid \Ft{t-1}) | &\leq \big\{\V(D_t\mid\Ft{t-1}) \V(C_t\mid\Ft{t-1})\big\}^{1/2} \\
        &\lesssim \big\{\bar R^2 R(K) (\bar\sigma^{2}+4\bar R^{2})(h\varepsilon_t)^{-1}\big\}^{1/2}\\
        &= O\big((h\varepsilon_t)^{-1/2}\big)
    \end{align*}

    \paragraph{Combining.} Combining all of the bounds, we find
    \begin{align*}
      \sigma_t^2 &= \underbrace{\V(D_t\mid\Ft{t-1})}_{O(1)} +\underbrace{2 Cov(D_t,C_t\mid\Ft{t-1})}_{O((h\varepsilon_t)^{-1/2})} +\underbrace{\V(C_t\mid\Ft{t-1})}_{h^{-1}\bar V_t+O(1)} \\
      &= h^{-1}\bar V_t+O_p\big((h\varepsilon_t)^{-1/2}\big)
    \end{align*}
    The lower and upper bound follows from our bound on $\E[C_t^2]$. 
    
    \subsubsection{Variance of the Process}

    Since $w_t$ is $\Ft{t-1}$-measurable and $\E[\Delta_t\mid\Ft{t-1}]=0$ by part (a), $\{w_t\Delta_t\}$ is a martingale difference sequence and cross-terms cancel, meaning for $\E[\Delta_t^{2}\mid\Ft{t-1}]=\sigma_t^{2}$,
    \[
      \Sigma_T=\E\Big[\Big(\sum_{t\le T}w_t\Delta_t\Big)^{2}\Big]
      =\E\Big[\sum_{t\le T}w_t^{2}\Delta_t^{2}\Big]
      =\sum_{t\le T}\E\big[w_t^{2}\sigma_t^{2}\big]
    \]

    It remains to bound each $\E[w_t^{2}\sigma_t^{2}]$. 
    First, note the remainder term sums the following terms, almost surely:
    \[
      r_t = \underbrace{\V(D_t\mid\Ft{t-1})}_{\le \bar R^2}
          - \underbrace{\big(\E[C_t\mid\Ft{t-1}]\big)^2}_{\le 4 \bar R^2}
          + \underbrace{2\,\mathrm{Cov}(D_t,C_t\mid\Ft{t-1})}_{\le 2\bar R \frac{C_V}{(h\varepsilon_t)^{1/2}}}
    \]
    Hence, let $C_r=5\bar R^2+2\bar R C_V^{1/2}$ so that $|r_t|\le C_r(h\varepsilon_t)^{-1/2}$. Applying the bound from (b), $c_V\le\bar V_t\le C_V/\varepsilon_t$ from above to $\sigma_t^{2}=h^{-1}\bar V_t+r_t$,
    \begin{align*}
      \sigma_t^{2} & \ge \frac{c_V}{h}-C_r(h\varepsilon_t)^{-1/2} = \frac{c_V}{h}\left(1-\frac{C_r}{c_V}\cdot\frac{h}{(h\varepsilon_t)^{1/2}}\right)  = \frac{c_V}{h}\left(1-\frac{C_r}{c_V}\Big(\frac{h}{\varepsilon_t}\Big)^{1/2}\right) \\
      \sigma_t^{2} & \le \frac{C_V}{h\varepsilon_t}+C_r(h\varepsilon_t)^{-1/2} = \frac{C_V}{h\varepsilon_t}\left(1+\frac{C_r}{C_V}\cdot\frac{h\varepsilon_t}{(h\varepsilon_t)^{1/2}}\right) = \frac{C_V}{h\varepsilon_t}\left(1+\frac{C_r}{C_V}(h\varepsilon_t)^{1/2}\right)
    \end{align*}

    Now we simplify this notation using a couple of bounds. 
    Let $C_p=C_r/c_V$. Since $c_V\le C_V$, $C_r/C_V \le C_r/c_V$. We can bound $C_V=R(K)(\bar\sigma^2+4\bar R^2)\ge4R(K)\bar R^2$, so $\bar R\le(C_V/4R(K))^{1/2}$, implying
    \[
      C_r = 5\bar R^2+2\bar R C_V^{1/2}
       \le \frac{5C_V}{4R(K)}+\frac{C_V}{R(K)^{1/2}}
       = C_V\Big(\frac{5}{4R(K)}+\frac{1}{R(K)^{1/2}}\Big)
    \]
    The result then follows by summing.
\end{proof}

\subsection{\texorpdfstring{Proof of \cref{thm:mse}}{}}
\label{pf:mse}

\begin{proof}
  \listsafesectioning
Throughout, we will use the total expectation on the event that 
\labelcref{lemma:weightReg}(ii) holds, that is, 
$$ \mathcal{E}_T = \{ |\frac{\sum_t w_t}{\E[\sum_t w_T]}-1|\leq1/2 \} $$
By Markov's, $\pr(\mathcal{E}_T^C)\leq 2 \E[|\frac{\sum_t w_t}{\E[\sum_t w_t]}-1|] \lesssim \delta_T$ so it occurs with high probability.

Subsequently, the proof proceeds as follows. First
\begin{align*}
    \mathrm{MSE}
    &= \E\big[(\hat{V}_T - V)^2(\policy)\big] \\
    &= \E\big[(\hat{V}_T - V_h)^2(\policy)\big]
       + 2 \E\big[(V_h-V)(\policy)(\hat{V}_T-V_h)(\policy)\big]
       + (V_h-V)^2(\policy) \\
    &= \underbrace{\E\left[\frac{\big(\sum_tw_t\Delta_t\big)^2}{\big(\sum_tw_t\big)^2}\right]}_{\textup{(I)}}
       + \underbrace{2(V_h-V)(\policy) \E\left[\frac{\sum_tw_t\Delta_t}{\sum_tw_t}\right]}_{\textup{(II)}}
       + (V_h-V)(\policy)^2 \\
    &= \frac{\Sigma_T}{\big(\E\sum_t w_t\big)^2}\big(1+O(\delta_T)\big)
       + O\left(h^2\delta_T^{1/2} \frac{\Sigma_T^{1/2}}{\E\sum_tw_t}\right)
       + O(h^4)
\end{align*}
where the third equation follows as $\hat V_T-V_h=\frac{\sum_t w_t \Delta_t }{\sum_t w_t}$ and the bias result on $(V_h-V)(\policy)=O(h^2)$ from \cref{lemma:biasVariance}(a). We claim the final line, and demonstrate below.

\paragraph{Term \textup{(I)}.} 
We claim 
$$ \textup{(I)} = \frac{\Sigma_T}{\big(\E\sum_tw_t\big)^2}\big(1+O(\delta_T)\big)$$
We show this by demonstrating
\[
  \left|\E\left[\frac{\big(\sum_tw_t\Delta_t\big)^2}{\big(\sum_tw_t\big)^2}\right]
    -\frac{\Sigma_T}{\big(\E\sum_tw_t\big)^2}\right|
   \le O(\delta_T) \frac{\Sigma_T}{\big(\E\sum_tw_t\big)^2}
\]

From $\textup{(I)}$, multiply and divide by $\big(\E\sum_tw_t\big)^2$, and consider
\[
  \frac{\big(\sum_tw_t\Delta_t\big)^2}{\big(\sum_tw_t\big)^2} -\frac{\big(\sum_tw_t\Delta_t\big)^2}{\big(\E\sum_tw_t\big)^2}
  =\frac{\big(\sum_tw_t\Delta_t\big)^2}{\big(\E\sum_tw_t\big)^2} \left[\frac{\big(\E\sum_tw_t\big)^2}{\big(\sum_tw_t\big)^2}-1\right]
\]

We handle on $\mathcal{E}_T$ and on $\mathcal{E}_T^C$ as separate cases below.

\textbf{On $\mathcal{E}_T$:}\\
On $\mathcal{E}_T$ the sum of weights behaves well as
$\tfrac12 \E[\sum_t w_t] \leq \sum_t w_t \leq \tfrac32 \E[\sum_t w_t]$. Use this to simplify the bracketed expression
\begin{align*}
    \left|\frac{\big(\E[\sum_t w_t]\big)^2}{(\sum_t w_t)^2}-1\right|
    &= \frac{\big|\E[\sum_t w_t]-\sum_t w_t\big| \cdot \big|\E[\sum_t w_t]+\sum_t w_t\big|}{(\sum_t w_t)^2}\\
    &\lesssim \frac{\big|\E[\sum_t w_t]-\sum_t w_t\big| \cdot \E[\sum_t w_t]}{\big(\E \sum_t w_t\big)^2} \\
    &\lesssim \left|\frac{\sum_t w_t}{\E[\sum_t w_t]}-1\right|
\end{align*}
using $|\E\sum_t w_t + \sum_t w_t| \leq \tfrac52\E\sum_t w_t$ and
$(\sum_t w_t)^2 \geq \tfrac14 (\E\sum_t w_t)^2$ in the second line.

Multiply by $\frac{(\sum_t w_t \Delta_t)^2}{\big(\E\sum_t w_t\big)^2}$ to recover the primary expression of interest. Since the bracket is at most $\tfrac12$ in absolute value on $\mathcal{E}_T$, Cauchy--Schwarz with
$\E[(\sum_t w_t\Delta_t)^4]\lesssim\Sigma_T^2$ gives
\begin{align*}
    \left|\E\left[\frac{(\sum_t w_t \Delta_t)^2}{\big(\E\sum_t w_t\big)^2}
      \left(\frac{\big(\E[\sum_t w_t]\big)^2}{(\sum_t w_t)^2}-1\right)\I{\mathcal{E}_T}\right]\right|
    &\lesssim \frac{\big(\E[(\sum_t w_t\Delta_t)^4]\big)^{1/2}}{\big(\E\sum_t w_t\big)^2} \left(\E\left|\frac{\sum_t w_t}{\E[\sum_t w_t]}-1\right|\right)^{1/2} \\
    &\lesssim \frac{\Sigma_T}{\big(\E \sum_t w_t\big)^2} \delta_T^{1/2}
\end{align*}

\textbf{On $\mathcal{E}_T^C$:}\\
Note the estimator is a weighted average of the martingale increments: 
\begin{align*}
    \left|\frac{\sum_t w_t\Delta_t}{\sum_t w_t}\right|
    \leq \sum_t \frac{w_t}{\sum_s w_s}|\Delta_t|
    \leq \max_{t\le T}|\Delta_t|
    \leq \max_{t\le T} M_t 
\end{align*}
Combined with Markov's inequality, $\pr(\mathcal{E}_T^C) \leq 2\delta_T$, yielding
\begin{align*}
    \left|\E\left[\frac{(\sum_t w_t\Delta_t)^2}{(\sum_t w_t)^2}\I{\mathcal{E}_T^C}\right]\right| &\leq \big(\max_{t\le T} M_t\big)^2 \pr(\mathcal{E}_T^C) \\
    &\leq 2\delta_T \big(\max_{t\le T} M_t\big)^2 \\
    &= o_p\left(\frac{\Sigma_T}{\big(\E\sum_t w_t\big)^2}\right)
\end{align*}

\textbf{Conclude on I:}
Combining above,
\begin{align*}
    \textup{(I)} = \frac{\Sigma_T}{\big(\E\sum_t w_t\big)^2}\big(1+O(\delta_T^{1/2})\big)
\end{align*}


\paragraph{Term \textup{(II)}.} 
$\E\big[\sum_tw_t\Delta_t\big]=0$ and $\E[\sum_tw_t]$ is a constant, so
\[
  \E\left[\frac{\sum_tw_t\Delta_t}{\sum_tw_t}\right]
  =\E\left[\Big(\sum_tw_t\Delta_t\Big)\left(\frac{1}{\sum_tw_t}-\frac{1}{\E\sum_tw_t}\right)\right]
\]

\textbf{On $\mathcal{E}_T$}\\
\begin{align*}
    |\frac{1}{\sum_t w_t}- \frac{1}{\E[\sum_t w_t]}| &= \frac{|\E[\sum_t w_t]-\sum_t w_t|}{\sum_t w_t)(\E \sum_t w_t)} \\
    &\leq 2\frac{1}{\E[\sum_t w_t]}|\frac{\sum_t w_t}{\E[\sum_t w_t]}-1|
\end{align*}
Now, multiplying in $\sum_t w_t \Delta_t$, 
$$ |\E[\sum_t w_t \Delta_t (\frac{1}{\sum_t w_t}- \frac{1}{\E[\sum_t w_t]}) \I{\mathcal{E}_T}]| \leq \frac{\Sigma_T^{1/2}}{\E[\sum_t w_t]} \delta_T^{1/2}$$

\textbf{On $\mathcal{E}_T^C$}\\
By \labelcref{ass:exp}, $\pi_t(\policy\mid s)\geq\varepsilon_t$, so $\E_S[\pi_t(\policy\mid S)^{-1}]\leq\varepsilon_t^{-1}$. Since $g(x)=x^{-\gamma}$ is decreasing,
$w_t\geq\varepsilon_t^{\gamma}$ surely. Summing,
\[
  \sum_t w_t \geq \sum_t \varepsilon_t^{\gamma} = \sum_t \underline{w}_t
\]
Write $\kappa_T = \E[\sum_t w_t]/(\sum_t \underline{w}_t) \geq1$. Then, by Cauchy--Schwarz and $\pr(\mathcal{E}_T^C)\leq2\delta_T$,
\begin{align*}
    \left|\E\left[\frac{(\sum_t w_t\Delta_t)^2}{(\sum_t w_t)^2}\I{\mathcal{E}_T^C}\right]\right|
    &\leq \frac{1}{(\sum_t \underline{w}_t)^2}\E\left[\Big(\sum_t w_t\Delta_t\Big)^2\I{\mathcal{E}_T^C}\right] \\
    &\leq \frac{1}{(\sum_t \underline{w}_t)^2}
      \Big(\E\Big[\big(\sum_t w_t\Delta_t\big)^4\Big]\Big)^{1/2}
      \pr(\mathcal{E}_T^C)^{1/2} \\
    &\lesssim \kappa_T^2 \frac{\Sigma_T}{\big(\E\sum_t w_t\big)^2} \delta_T^{1/2}
\end{align*}
using $\E[(\sum_t w_t\Delta_t)^4]\lesssim\Sigma_T^2$ and $(\sum_t \underline{w}_t)^{-2}=\kappa_T^2(\E\sum_t w_t)^{-2}$.
This is $o\big(\Sigma_T/(\E\sum_t w_t)^2\big)$ whenever $\kappa_T^4\delta_T\to0$. 
\\\\
\textbf{Conclude on II:}\\
Then, combine to conclude
$$ | II | = 2 |(V-V_h)(\policy)| = O(h^2 \delta_T^{1/2} \frac{\Sigma_T^{1/2}}{\E[\sum_t w_t]})$$

Note,  $h^2\frac{\Sigma_T^{1/2}}{\E\sum_tw_t} \leq \tfrac12\big(h^4+(\frac{\Sigma_T^{1/2}}{\E\sum_tw_t})^2\big)$ for any $h$.  Multiplying by $\delta_T^{1/2}$ so
\[
  |\mathrm{II}|\lesssim\delta_T^{1/2}\Big(\frac{\Sigma_T}{(\E\sum_tw_t)^2}+h^4\Big)
\]
Since $\delta_T\to0$ by \cref{lemma:weightReg}, term \textup{(II)} is negligible relative to the variance and bias.

\end{proof}

\subsection{\texorpdfstring{Proof of \cref{lemma:varEst}}{}}
\label{pf:varEst}
\begin{proof}
  \listsafesectioning
    \subsubsection{Part 1: Oracle Variance Consistency}
    We will show
    \begin{align*}
      \E \left| \frac{\sum_t w_t^2 \sigma_t^2}{\Sigma_T} - 1 \right| \lesssim \E \frac{ \sum_t c_t\big(|\psi_t| + \E|\psi_t|\big)}{\sum_t c_t}
    \end{align*}
    for deterministic $c_t>0$ with $\sum_{t\le T} c_t \to \infty$ and random $\psi_t \xrightarrow{a.s.} 0$ with $| \psi_t | < \infty$ a.s. The result follows by Toeplitz's lemma.

    First, we simplify $\sigma_t^2$ by rewriting $\bar V_t$. 
    Recall from \cref{lemma:biasVariance}, $\bar V_t =  R(K) \E[\frac{m_2(S,A)}{\pi_t(\policy \mid S)}]$ where $m_2(s,a) = \sigma_R^2(s,a) + (\mu - \hat{\mu}_{t-1})^2(s,a)$, and the following identities:
    \begin{align*}
        \Big|\frac{1}{\pi_t(\policy\mid s)} - \E\Big[\frac{1}{\pi_t(\policy\mid s)}\Big]\Big| &\leq \rho_t \E\Big[\frac{1}{\pi_t(\policy\mid s)}\Big], \qquad \mbox{for }\rho_t\mbox{ as in }\labelcref{ass:suffCoverage}\\
        \big| m_2(s,a) - \sigma_R^2(s,a) \big| &\leq \lambda_t\, \sigma_R^2(s,a), \qquad \lambda_t = \frac{\|\mu-\hat{\mu}_{t-1}\|_\infty^2}{\underline{\sigma}^2} 
    \end{align*}
    Now, rewrite $\sigma_t^2$ plugging in the above identities 
        \begin{align*}
        \sigma_t^2 &= h^{-1}\bar V_t + r_t \\
        &= h^{-1} R(K) \E_S \left[\sigma_R^2(S,\policy) \E \left[\frac{1}{\pi_t(\policy\mid S)}\right]\right] \big(1+O(\rho_t)\big)\big(1+O(\lambda_t)\big) + O\big((h\varepsilon_t)^{-1/2}\big) \\
        &= h^{-1} R(K) \E_S \left[\sigma_R^2(S,\policy) \E \left[\frac{1}{\pi_t(\policy\mid S)}\right]\right] \Big(1+O\big(\rho_t+\lambda_t+\sqrt{h/\varepsilon_t} \big)\Big) \\
        &= h^{-1} C_t(K) [1+O(\theta_t)]
    \end{align*}
    where we have taken $C_t(K) =  R(K)\E[\sigma_R^2(s,a)\frac{1}{\pi_t(\policy\mid s)}] $ and $\theta_t =\rho_t +\lambda_t + \sqrt{h \varepsilon_t^{-1}}$. 

    The weights $w_t$ are also random, so we require a substitution before simplifying in the full expression. Weights in consideration are some power function $g\big(\E_S[1/\pi_t(\policy\mid S)]\big)$ of the averaged inverse density. By \cref{lemma:weightReg},
    \begin{align*}
        w_t &= \bar w_t\big(1+O(\rho_t)\big), \qquad
        \bar w_t = g\left(\E_S\left[\E\left[\frac{1}{\pi_t(\policy\mid S)}\right]\right]\right) \\
        \implies\ &  w_t^2=\bar w_t^2\big(1+O(\rho_t)\big)
    \end{align*}

    Finally, we plug in, for $c_t = \bar w_t^2 C_t(K)$ and $\psi_t = O(\rho_t+\theta_t)  = O(\theta_t)$,
    \begin{align*}
        w_t^2\sigma_t^2 &= \bar w_t^2\big(1+O(\rho_t)\big)\cdot h^{-1}C_t(K)\big(1+\theta_t\big) \\
        &= h^{-1}c_t(K)\big(1+\psi_t\big)
    \end{align*}
    Simplifying the primary expression of interest with these identitites, we find 
    \begin{align*}
        \left| \frac{\sum_t w_t^2\sigma_t^2}{\Sigma_T} - 1 \right| &=\left| \frac{h^{-1}\sum_t c_t(1+\psi_t)}{h^{-1}\sum_t c_t\big(1+\E\psi_t\big)} - 1 \right|\\
        &= \left|\frac{\sum_t c_t(\psi_t - \E\psi_t)}{\sum_t c_t(1+\E\psi_t)} \right|\\
        & \leq \frac{\sum_t c_t (| \psi_t| + \E|\psi_t|) }{ \sum_t c_t (1- \sup_t \E|\psi_t| )}
    \end{align*}
   Since $\sup_t \E |\psi_t| \rightarrow 0$ by \labelcref{ass:suffCoverage} and convergence of our nuisance, we have our initial claim, thereby demonstrating the result upon applying Toeplitz.
    \begin{align*}
        \E \left|\frac{\sum_t w_t^2\sigma_t^2}{\Sigma_T} - 1\right|
        \lesssim \E \frac{\sum_t \bar c_t\big(|\psi_t| + \E|\psi_t|\big)}{\sum_t \bar c_t}
    \end{align*}

    \subsubsection{Part 2: Consistent Variance Estimation}
    
    Decompose the ratio as follows:
    \begin{align*}
        \frac{\hat{\Sigma}_T}{\Sigma_T} &= \underbrace{\frac{\sum_t w_t^2 \sigma_t^2}{\Sigma_T}}_{\textup{(I)}}
        + \underbrace{\frac{\sum_t w_t^2 (\Delta_t^2 - \sigma_t^2)}{\Sigma_T}}_{\textup{(II)}}
        - 2 \underbrace{(\hat{V}_T - V_h)\frac{\sum_t w_t^2 \Delta_t}{\Sigma_T}}_{\textup{(III)}}
        +  \underbrace{(\hat{V}_T - V_h)^2\frac{\sum_t w_t^2}{\Sigma_T}}_{\textup{(IV)}}
    \end{align*}
    We handle each term by term, showing 
    \begin{align*}
        I &\xrightarrow{p} 1 \\
        II &= o_p(1) \\
        III&= o_p(1) \\
        IV &= o_p(1) \\
    \end{align*}
    Throughout, note that $\hat{V}_T -V_h = \sum_t w_t \Delta_t / \sum_t w_t$.

    \paragraph{Term {I}:}
    By part 1, $I \xrightarrow{p} 1$. 

    \paragraph{Term II:}
    We will invoke Lenglart's inequality \citep{lenglart1977relation} to demonstrate $\mathrm{II}=o_p(1)$. To that end, we bound the conditional variance of the numerator as follows. The summands $\xi_t=w_t^2(\Delta_t^2-\sigma_t^2)$ form a martingale difference sequence, since $w_t$ is $\Ft{t-1}$-measurable and $\E[\Delta_t^2\mid\Ft{t-1}]=\sigma_t^2$. Recalling $|\Delta_t| \leq M_t$, we have
    \begin{align*}
      \frac{1}{\Sigma_T^2} \sum_t\E\big[\xi_t^2\mid\Ft{t-1}\big] & \leq \frac{1}{\Sigma_T^2} \sum_t w_t^4\E\big[\Delta_t^4\mid\Ft{t-1}\big] \\
      &\leq \frac{1}{\Sigma_T^2}\sum_t w_t^4 M_t^2 \E\big[\Delta_t^2\mid\Ft{t-1}\big] \\
      &\leq \max_t w_t^2 M_t^2 \frac{ \sum_t w_t^2\sigma_t^2}{\Sigma_T^2} \\
      &= o_p(1)\cdot O_p(1)\\
      & =o_p(1)
    \end{align*}
    by \labelcref{ass:maxWeightRegularity} for the first factor and the first claim for the second.
    By Lenglart's inequality \citep{lenglart1977relation}, for any $a,\eta>0$,
    \[
      \pr\big(|\mathrm{II}|>a\big) =\pr\Big(\Big|\sum_t\xi_t\Big|>a\Sigma_T\Big)
      \leq \frac{\eta}{a^2}+\pr\Big(\frac{1}{\Sigma_T^2}\sum_t\E\big[\xi_t^2\mid\Ft{t-1}\big]>\eta\Big)
    \]
    The second term vanishes as $T\to\infty$, so $\limsup_T\pr(|\mathrm{II}|>a)\le\eta/a^2$. Since $\eta$ is arbitrary, $\mathrm{II}=o_p(1)$.

    \paragraph{Term IV:}
    Noting $\sum_t w_t \Delta_t = O_p(\Sigma^{1/2}_T)$ we have that
    \begin{align*}
        IV &= (\frac{\sum_t w_t \Delta_t}{\sum_t w_t})^2 \frac{\sum_t w_t^2}{\Sigma_T} \\
        &= O_p(\frac{\Sigma_T}{(\sum_t w_t^2)^2} \frac{\sum_t w_t^2}{\Sigma_T}) \\
        &= o_p(1)
    \end{align*}

    \paragraph{Term III:} This can be analyzed as follows
    \begin{align*}
        III &= \frac{\sum_t w_t \Delta_t}{\sum_t w_t}\cdot\frac{\sum_t w_t^2 \Delta_t}{\Sigma_T} \\
        &\le \frac{\sum_t w_t \Delta_t}{\sum_t w_t}\cdot \frac{\big(\sum_t w_t^2\big)^{1/2}\big(\sum_t w_t^2\Delta_t^2\big)^{1/2}}{\Sigma_T}, \quad \mbox{by Cauchy--Schwarz on the second term} \\
        &= \left(\frac{\sum_t w_t^2\Delta_t^2}{\Sigma_T}\right)^{1/2} \left(\Big(\frac{\sum_t w_t\Delta_t}{\sum_t w_t}\Big)^{2}\frac{\sum_t w_t^2}{\Sigma_T}\right)^{1/2} \\
        & = (I + II)^{1/2} (IV)^{1/2} \\
        &= o_p(1)
    \end{align*}
    where the last step leverages our analyses from I, II, and IV.
\end{proof}

\subsection{\texorpdfstring{Proof of \cref{thm:ptNormality}}{}}
\label{pf:ptNormality}
\begin{proof}
  \listsafesectioning
    Here we invoke the martingale CLT \citep[Thm.~3.2, Cor.~3.1]{hallheyde80}. Specifically, recall that almost surely
    \begin{align*}
        \Delta_t &= \hat\Gamma_t - V_h(\policy), \quad |\Delta_t| \leq M_t \\
        \Sigma_T &= \E[\sum_{t\le T}w_t^2\sigma_t^2]
    \end{align*}
    Consider the decomposition using $\hat{V}_T(\policy)-V_h(\policy)=\frac{\sum_t w_t \Delta_t}{\sum_t w_t}$,
    \begin{align*}
      \frac{\Vhat(\policy)-V(\policy)}{\hat\Sigma_T^{1/2}/\sum_tw_t}
      &=\big(1+k_T\big)Z_T
    \end{align*}
    where 
    \begin{align*}
        k_T&=\frac{\Sigma_T^{1/2}}{\hat\Sigma_T^{1/2}}-1 \\
        Z_T&=\underbrace{\frac{\sum_tw_t\Delta_t}{\Sigma_T^{1/2}}}_{v_T}
      +\underbrace{\frac{\big(\sum_tw_t\big)b_h(\policy)}{\Sigma_T^{1/2}}}_{b_T}
    \end{align*}     
    We will show $Z_T \rightsquigarrow \N(0,1)$ by showing  $v_T\rightsquigarrow\N(0,1)$ and $b_T\to_p 0$. Further, we claim by \cref{lemma:varEst}, $k_T = o_p(1)$, yielding the result.

    \paragraph{Distribution term $v_T$.} Define the martingale difference array $\xi_{t,T}=\frac{w_t \Delta_t}{\Sigma_T^{1/2}}$. 
    \begin{itemize}
      \item[(MG1)] \emph{Predictable variance.}
        
        Since $V_h(\policy)=\E[\hat\Gamma_t(\policy)\mid\Ft{t-1}]$ and $w_t$ is $\Ft{t-1}$-measurable,
        \begin{align*}
        \sum_{t\le T}\E\big[\xi_{t,T}^2 \mid \Ft{t-1}\big] &= \frac{1}{\Sigma_T}\sum_{t\le T}w_t^2 \E\big[(\hat\Gamma_t(\policy)-V_h(\policy))^2\mid\Ft{t-1}\big] \\
          &= \frac{1}{\Sigma_T}\sum_{t\le T}w_t^2 \V\big(\hat\Gamma_t(\policy)\mid\Ft{t-1}\big) \\
          &= \frac{1}{\Sigma_T} \sum_t w_t^2 \sigma_t^2 
        \end{align*}        

        This converges to 1 by \cref{lemma:varEst}.

      \item[(MG2)] \emph{Conditional Lindeberg.} 

        First recall our bound from \cref{lemma:biasVariance} on the martingale increment: $|\Delta_t| \le M_t = \frac{2\kappa\bar R}{h \varepsilon_t}+2\bar R$. This followed as on the kernel's support $|A_t-\policy(s_t)|\le h \le \delta_0$ we have, $|K_h|\le \frac{\kappa}{h}$, the importance weight is at most $1/\varepsilon_t$ by \labelcref{ass:exp}, and $|R_t-\hat\mu_{t-1}| \le 2 \bar R$ by \labelcref{ass:bdd}. 
        
        Now we demonstrate the Lindeberg Condition. Since $|\xi_{t,T}|\le \frac{w_t M_t}{\Sigma_T^{1/2}}$, $\{|\xi_{t,T}|>\delta\} \subseteq \{ \frac{w_t M_t}{\Sigma_T^{1/2}} > \delta\}$. Thus,
        \begin{align*}
          \sum_{t\le T}\E\big[\xi_{t,T}^2\I{|\xi_{t,T}|>\delta}\mid\Ft{t-1}\big] & \leq \sum_{t \leq T} \I{\frac{w_t M_t}{\Sigma_T^{1/2}} > \delta}\E\big[\xi_{t,T}^2\mid\Ft{t-1}\big] \\
            &\le \I{\frac{\max_t w_t M_t}{\Sigma_T^{1/2}} > \delta}\sum_{t\le T}\E\big[\xi_{t,T}^2\mid\Ft{t-1}\big] \\
            &= o_p(1) O_p(1) \\
            &= o_p(1)
        \end{align*}
        where the second to last line follows as the first expression is $o_p(1)$ by \labelcref{ass:maxWeightRegularity} and the second expression is $O_p(1)$ by (MG1).
    \end{itemize}

    \paragraph{Bias term $b_T$.} Since $b_h=V_h-V$ exactly, $Z_T-b_T=v_T\rightsquigarrow\N(0,1)$, and multiplying by $1+k_T=1+o_p(1)$ gives the first statement by Slutsky's. 
    
    Under \labelcref{ass:underSmooth},  
    \begin{align*}
        |b_T| &= \frac{\big(\sum_t w_t\big)\big|\underbrace{V_h(\policy)-V(\policy)}_{b_h(\policy)}\big|}{\Sigma_T^{1/2}} \\
        &\lesssim \frac{h^2 \sum_t w_t}{\Sigma_T^{1/2}},
        \quad \mbox{by the bias calculation in } \cref{lemma:biasVariance}\textup{(a)} \\
        &= o_p(1), \quad \mbox{by }\labelcref{ass:underSmooth}
    \end{align*}
    So we may omit the bias term.

\end{proof}

    
    

\paragraph{Lower bounds.}
Lower bounds depend on how many observations the design collects near the target policy; the lemma below will make this precise by characterizing the mutual information in a continuous action setting. Then, we will use this to derive lower bounds for mean squared error.

Fix a target policy $\pi$ and the context law $P_S$. Instances are indexed by the mean reward $\mu$, and we work over the class
\[
  \mathcal P(L,\nu)=\Big\{P_\mu:R\mid(S,A)\sim\mathcal N\big(\mu(S,A),\nu^2\big),\sup_{s,a}|\partial_a^2\mu(s,a)|\le L\Big\}
\]
of Gaussian reward models whose mean has bounded curvature in $a$. The value of $\pi$ is $V(\policy)=\E_S\mu(S,\policy)$. With $\pi^*(s)=\arg\max_a\mu(s,a)$ the optimal policy, the gap of action $a$ in context $s$ is $\zeta(s,a)=\mu(s,\pi^*(s))-\mu(s,a)\ge0$, and the regret is $\mathrm{Reg}_T=\sum_{t\le T}\zeta(S_t,A_t)$. Let $\mathcal N_{\policy}(h)=[\policy-h,\policy+h]$ be the $h$-window of the target and
\[
  N_h=\#\{t\le T:A_t\in\mathcal N_{\policy(s_t)}(h)\}
\]
the number of observations in it.

\begin{lemma}\label{lemma:KLD}
Let $\phi\in C^2$ satisfy $\phi(0)=1$ and $\operatorname{supp}\phi\subseteq[-1,1]$. Let $P_0\in\mathcal P(L/2,\nu)$ have mean $\mu_0$, and for $h>0$, $0<\lambda\le 1/(2\|\phi''\|_\infty)Lh^2$, define $P_1$ by
\[
  \mu_1(s,a)=\mu_0(s,a)+\lambda\phi\Big(\frac{a-\policy}{h}\Big)
\]
Then $P_1\in\mathcal P(L,\nu)$, and if the same algorithm is run under $P_0$ and $P_1$, we have satisfy
\[
  KL(P_0\|P_1)\le\frac{\lambda^2}{2\nu^2}\E_0[N_h]
\]
\end{lemma}

\begin{proof}
    First, we address the claim $P_1\in\mathcal P(L,\nu)$. This follows as the second derivative in $a$ is indeed bounded by $L$. Write $u=(a-\policy)/h$, so that $du/da=h^{-1}$. By the chain rule,
    \[
      \partial_a\big[\lambda\phi(u)\big]=\lambda h^{-1}\phi'(u)
      \qquad
      \partial_a^2\big[\lambda\phi(u)\big]=\lambda h^{-2}\phi''(u)
    \]
    so the curvature of the perturbation satisfies
    \[
      \big|\lambda h^{-2}\phi''(u)\big|\le\lambda h^{-2}\|\phi''\|_\infty\le 1/(2\|\phi''\|_\infty)Lh^2\cdot h^{-2}\|\phi''\|_\infty=\frac{L}{2\|\phi''\|_\infty}\|\phi''\|_\infty=\frac L2
    \]
    using the bound on $\lambda$ that was assumed. Since $|\partial_a^2\mu_0|\le L/2$, it follows that $|\partial_a^2\mu_1|\le L/2+L/2=L$, yielding the claim $P_1\in\mathcal P(L,\nu)$.
    
    Second, we compute the KL bound.
    \begin{align*}
        KL(P_0 \| P_1) &= \sum_t \E_0 [\frac{(\mu_0 - \mu_1)^2 (S_t,A_t)}{2\nu^2}] \\
        &= \sum_t \E_0[\frac{\lambda^2 \phi^2 ( \frac{A_t - \policy(s_t)}{h})}{2 \nu^2}] \\
        &\leq \frac{\lambda^2 }{2\nu^2} \sum_t \pr_0(A_t \in \mathcal{N}_{\policy(s_t)}(h)) \\
        & = \frac{\lambda^2 }{2\nu^2} \E_0 [N_h]
    \end{align*}
    where the inequality in the third line follows from $\phi \in [-1,1]$, and only evaluating if the event in question $A_t \in \mathcal{N}_{\policy(s_t)}(h)$ happens.
\end{proof}

\subsection{\texorpdfstring{Proof of \cref{thm:lb}}{}}
\label{pf:lb}
    We prove the following result. Write $g^*_0(s)=\arg\max_a\mu_0(s,a)$, so that $\mathrm{Regret}_T=\sum_{t\le T}[\mu_0(S_t,g^*_0(S_t))-\mu_0(S_t,A_t)]$, and let $M=\sup_{P\in\mathcal P(L,\nu)}\E_P[(\hat V-V_P(\policy))^2]$.
    \begin{enumerate}[label=\textup{(\alph*)},leftmargin=1.6em]
      \item If $h=\big(4\|\phi''\|_\infty^2\nu^2/(L^2\E_0G)\big)^{1/5}\le\bar h$, then
      \[
        \inf_{\hat V}M\ge\frac{1}{16e}\Big(\frac{L}{2\|\phi''\|_\infty}\Big)^{2/5}\Big(\frac{\nu^2}{\E_0G}\Big)^{4/5}.
      \]
      \item Suppose $|\mu_0(s,a)-\mu_0(s,\policy(s))|\le L_V|a-\policy(s)|$ and let $h_M=\big(16\|\phi''\|_\infty M^{1/2}/L\big)^{1/2}$. If $h_M\le\bar h$ and $L_Vh_M\le\underline\Delta_\policy/2$, then
      \[
        \E_0[\mathrm{Regret}_T]\,M\ge\frac{\log2}{32}\underline\Delta_\policy\nu^2 .
      \]
    \end{enumerate}
\begin{proof}
    \emph{Bound (a).}
    Let $P_1$ and $P_0$ be the distributions described in \cref{lemma:KLD} 
    with $h \leq \bar h$ and height $\lambda = \frac{Lh^2}{2 \| \phi '' \|_\infty}$. 
    Note, assuming we can show that the KL between $P_0, P_1$ is bounded by 1, we can show by using Le Cam's two point lemma combined with Markov's concludes
    \begin{align*}
        \sup_{P \in \mathcal{P}(L,\nu)} \E_P[(\hat V - V_P(\policy))^2] &\geq \max_{j \in \{0,1 \}} \E_{P_j}[(\hat{V}-V_{P_j}(\policy))^2] \\
        &\geq (\frac{\lambda}{2})^2 \frac{1}{4e} 
    \end{align*}

    It remains to show how to choose $h$ so $\mathrm{KL}(P_0\|P_1)\le1$. 
    We will bound $N_h$, and then $KL$. Note that 
    \[
        \E_0[N_h]=\sum_{t\le T}\E_0\Big[\E_S\int_{\mathcal N_{\policy}(h)}\pi_t(a\mid S)da\Big]
        \le2h\sum_{t\le T}\E_0\Big[\E_S\sup_{a\in\mathcal N_{\policy}(\bar h)}\pi_t(a\mid S)\Big]
        =2h\E_0[G]
    \]
    Substituting this and $\lambda=cLh^2$ into \cref{lemma:KLD},
    \[
      \mathrm{KL}(P_0\|P_1)\le\frac{\lambda^2}{2\nu^2}\E_0[N_h]
      \le\frac{c^2L^2h^4}{2\nu^2}\cdot2h\E_0[G]
      =\frac{c^2L^2\E_0[G]}{\nu^2}h^5
    \]

    Set this to 1, and solve for $h$ that satisfies the bounded KLD to conclude
    \begin{align*}
        \frac{4L^2\E_0[G]}{\| \phi'' \|^2_\infty\nu^2}h^5 &= 1\\
        \iff h &= (\frac{\| \phi'' \|^2_\infty\nu^2}{4 L^2\E_0[G]})^{1/5}
    \end{align*}

    Plugging in this $h$, we find
    \[
      \lambda=\frac{2}{\|\phi''\|^2}Lh^2=\frac{2}{\| \phi'' \|_\infty}L\Big(\frac{\| \phi'' \|^2_\infty\nu^2}{4L^2\E_0G}\Big)^{2/5}
    \]

    Substitute this into our lower bound above to conclude the result.
    \[
      \Big(\frac\lambda2\Big)^2\frac1{4e}=\frac{\lambda^2}{16e}
      =\frac1{16e}\Big(\frac{L}{2\|\phi''\|_\infty}\Big)^{2/5}\Big(\frac{\nu^2}{\E_0G}\Big)^{4/5}
    \]

    \emph{Proof of (b).} We first show precision forces the algorithm to sample near the target. For $P_0,P_1$ as above with height $\lambda$, the same argument as in (a) gives
    \[
      M\ge\max_{j\in\{0,1\}}\E_{P_j}\big[(\hat V-V_{P_j}(\policy))^2\big]\ge\Big(\frac\lambda2\Big)^2\frac14e^{-\mathrm{KL}(P_0\|P_1)} 
    \]
    Now, choose the height $\lambda=8M^{1/2}$. The bump $\lambda\phi((a-\policy(s))/h)$ has second derivative at most $\lambda\|\phi''\|_\infty/h^2$, so choose the width $h_M$ to make this $L/2$, keeping $P_1\in\mathcal P(L,\nu)$:
    \[
      \frac{8M^{1/2}\|\phi''\|_\infty}{h_M^2}=\frac L2
      \iff
      h_M=\Big(\frac{16\|\phi''\|_\infty M^{1/2}}{L}\Big)^{1/2}
    \]
    Substituting into the lower bound on $M$ above,
    \[
      M\ge\frac{64M}{16}e^{-\mathrm{KL}(P_0\|P_1)}=4Me^{-\mathrm{KL}(P_0\|P_1)},
    \]
    Hence, dividing by $M$, $1\ge4e^{-\mathrm{KL}(P_0\|P_1)} \implies \mathrm{KL}(P_0\|P_1)\ge\log4$.
    
    Combining this with \cref{lemma:KLD} and $\lambda^2=64M$, $2\log2\le\mathrm{KL}(P_0\|P_1)\le\frac{32M}{\nu^2}\E_0[N_{h_M}]$ so $\E_0[N_{h_M}]\ge\frac{\log2}{16}\frac{\nu^2}{M}$.
\end{proof}

\section{Additional Simulation Details}\label{app:sim}
 
This section gives the details behind the simulation study. Throughout, $T=5000$, $S_t\sim\mathrm{Unif}[-1,1]$, $\A=[0,1]$, $R_t=\mu(S_t,A_t)+\epsilon_t$ with $\epsilon_t\sim\mathcal N(0,0.5^2)$, and $K(v)=\tfrac34(1-v^2)\I{|v|\le1}$, for which $R(K)=\tfrac35$ and $\mu_2(K)=\tfrac15$.

\subsection{Outcome models and true values}\label{app:sim:dgp}
\emph{Quadratic.} $g^*(s)=\tfrac12+\tfrac s8$ and $\partial_a^2\mu\equiv-4$, giving $V(\policy)=\tfrac{49}{96}$ on-target and $0.265$ off-target in closed form. The interior smoothing bias is then $V_h-V=-0.4h^2$.

\emph{Bimodal.} With $\sigma=0.09$ the ridges do not overlap, so $g^*(s)=0.75$ for $s>0$ and $0.25$ for $s<0$, with $V(\policy)\approx0.950$ on-target and $0.105$ off-target. At the peak $\E_S[\partial_a^2\mu]\approx-117$, so the leading bias is $-11.7h^2$, or $-0.019$ at $h=0.04$.

\emph{RFF.} $f(s,a)=\sqrt{2/D}\sum_{j\le D}\theta_j\cos(\omega_{s,j}s/\ell_s+\omega_{a,j}a/\ell_a+\beta_j)$ with: $D=64$, $(\ell_s,\ell_a)=(0.6,0.12)$, $\omega,\theta\sim\mathcal N(0,1)$, $\beta_j\sim\mathrm{Unif}[0,2\pi]$, drawn once, and $\mu=0.3+0.35(f-\bar f)/\mathrm{sd}(f)$. Here $g^*$ and all remaining true values are computed on a fine grid of size 200k.

\subsection{Logging design}\label{app:sim:logger}
The logging density is given by $\pi_t(a\mid s)=\rho_t+(1-\rho_t)\I{a\in B_t(s)}/|B_t(s)|$ where $B_t(s)=[\hat g_t(s)-b,\hat g_t(s)+b]\cap[0,1]$, $b=0.08$. 
The outcome fit $\hat\mu_t$, its degree $d$, and the bandwidth are chosen by forward-chaining 5-fold cross-validation, so that $\hat\mu_{t-1}$, $\hat g_t$, $h_t$ and $\pi_t$ are $\Ft{t-1}$-measurable. Refitting is done every 100 rounds, across degrees $2$ to $8$.

\subsection{Proposed estimator as implemented}\label{app:sim:ours}
\emph{Estimator.} For target dose $c$ and bandwidth $h$, we smooth the policy with the Epanechnikov kernel truncated to $[0,1]$ and renormalized, $K_{h,c}(a)\propto K((a-c)/h)\I{a\in[0,1]}$. Taking a second-order expansion of $\mu$ in $a$ gives the smoothing bias
\begin{equation}\label{eq:gap}
  V_h(\policy)-V(\policy)=\E_S\Big[\partial_a\mu(S,\policy(S))\phi_1(\policy(S),h)+\tfrac12\partial_a^2\mu(S,\policy(S))\phi_2(\policy(S),h)\Big]+O(h^3),
\end{equation}
where $\phi_1,\phi_2$ are the first two moments of $A-c$ under $K_{h,c}$. 

Each round contributes the standard kernel doubly robust score
\begin{align*}
  \hat\Gamma_t &= \int K_{h_t,\policy(S_t)}(a)\hat\mu_{t-1}(S_t,a)da+\frac{K_{h_t,\policy(S_t)}(A_t)}{\pi_t(A_t\mid S_t)}\big(R_t-\hat\mu_{t-1}(S_t,A_t)\big)
\end{align*}
and rounds are combined with inverse-variance weights $w_t$, computed from data before round $t$. 

We have
\begin{align*}
  \hat V &= \frac{\sum_tw_t\hat\Gamma_t}{\sum_tw_t}
  & \widehat{\mathrm{se}}_0^2 &= \frac{\sum_tw_t^2(\hat\Gamma_t-\hat V)^2}{\big(\sum_tw_t\big)^2}
\end{align*}

\emph{Bias correction.} We estimate the bias \eqref{eq:gap} with a local quadratic fit and subtract it. Let $U_t=A_t-\policy(S_t)$ and fit $R_t$ on the features $x_t=(S_t^jU_t^k)_{0\le j,k\le2}$ by least squares with weights $K(U_t/\lambda)$. We take $\lambda=2h_T$. 
This gives $\hat\beta$ and $G=\sum_tK(U_t/\lambda)x_tx_t^\top$ yielding an estimator of $\hat b=L^\top\hat\beta$ with
\begin{align*}
  L_{jk} &= \frac{\sum_tw_tS_t^j\phi_k(\policy(S_t),h_t)}{\sum_tw_t}, k=1,2; \quad  L_{j0} = 0
\end{align*}
We report $\hat V_{\rm bc}=\hat V-\hat b$, with standard error $\widehat{\mathrm{se}}^2=\widehat{\mathrm{se}}_0^2+\widehat{\V}(\hat b)-2\widehat{\mathrm{Cov}}(\hat V,\hat b)$.

\subsection{Baselines and evaluation}\label{app:sim:base}
We compare against four baselines. All use the true logged propensities $\pi_t(A_t\mid S_t)$, and those that need an outcome model use the same polynomial fit $\hat\mu$ as our method. Throughout, each uses a common bandwidth grid $h\in\{0.05,0.07,0.10,0.14\}$ for 5-fold cross-validation when relevant.
\begin{itemize}
    \item \emph{Kallus--Zhou} (self-normalized kernel IPW). With kernel weights $Z_t=K_{h,\policy(S_t)}(A_t)/\pi_t(A_t\mid S_t)$,
    \begin{align*}
      \hat V &= \frac{\sum_tZ_tR_t}{\sum_tZ_t}
      & \widehat{\mathrm{se}}^2 &= \frac{\sum_tZ_t^2(R_t-\hat V)^2}{\big(\sum_tZ_t\big)^2}
    \end{align*}
    \item \emph{Zhan} (binned doubly robust). The dose range is split into $B$ bins of width $1/B$, and the target dose is replaced by its bin $I_{k(s)}$. The score is
    \begin{align*}
      \hat\Gamma_t &= \hat\mu(S_t,m_{k})+\frac{\I{A_t\in I_{k}}}{e_t(S_t)}\big(R_t-\hat\mu(S_t,m_{k})\big)
      & e_t(s) &= \frac{\varepsilon_t}{B}+(1-\varepsilon_t)\frac{|I_{k(s)}\cap B_t(s)|}{|B_t(s)|}
    \end{align*}
    where $k=k(S_t)$, $m_k$ is the bin midpoint and $e_t$ is the probability that the logger lands in the bin. Rounds are combined with the variance-stabilizing weights $\sqrt{\bar e_t}$, where $\bar e_t$ is $e_t$ averaged over states, and the standard error is the martingale form used for our method. $B\in\{4,5,7,10,14,20,28\}$ is selected by forward-chaining cross-validation.
    \item \emph{Kennedy} (cross-fitted kernel doubly robust). The rounds are split into five random folds that ignore time order. On each fold, $\hat\mu$ is fit on the other four, and the pseudo-outcome is
    \begin{align*}
      \hat\Gamma_t &= \int K_{h,\policy(S_t)}(a)\hat\mu(S_t,a)da+\frac{K_{h,\policy(S_t)}(A_t)}{\pi_t(A_t\mid S_t)}\big(R_t-\hat\mu(S_t,A_t)\big)
    \end{align*}
    with the degree of $\hat\mu$ (and $h$) chosen by five-fold cross-validation. The estimate is $\bar\Gamma$ with the i.i.d.\ standard error $\mathrm{sd}(\hat\Gamma)/\sqrt n$
    \item \emph{Direct plug-in.} $\hat V=n^{-1}\sum_t\hat\mu(S_t,\policy(S_t))$ with $\hat\mu$ fit on all data. Its variance adds the sampling variance of $\hat\mu(S_t,\policy(S_t))$ over rounds, divided by $n$, to the model's estimation variance. 
\end{itemize}

Each cell is replicated $1000$ times. 
 
\subsection{Comparison On-Policy, Jiggered}\label{sec:policyCmp}
We run a comparison to the optimal policy, with a constant offset of 0.04. The results are shown in \cref{fig:coverage-full} and \cref{fig:mse-full} with the optimal and an off-policy target from the main text. 

The on-policy $+0.04$ target, a policy shifted $0.04$ from the optimum, is a challenging target policy: its kernel window lies near the edge of the logger's band. Our method covers nominally for $\alpha \le 0.5$ on all outcome models, but undercovers in the bimodal setting though still superior to its competitors.

\begin{figure*}[t]
  \centering
  \includegraphics[width=\textwidth]{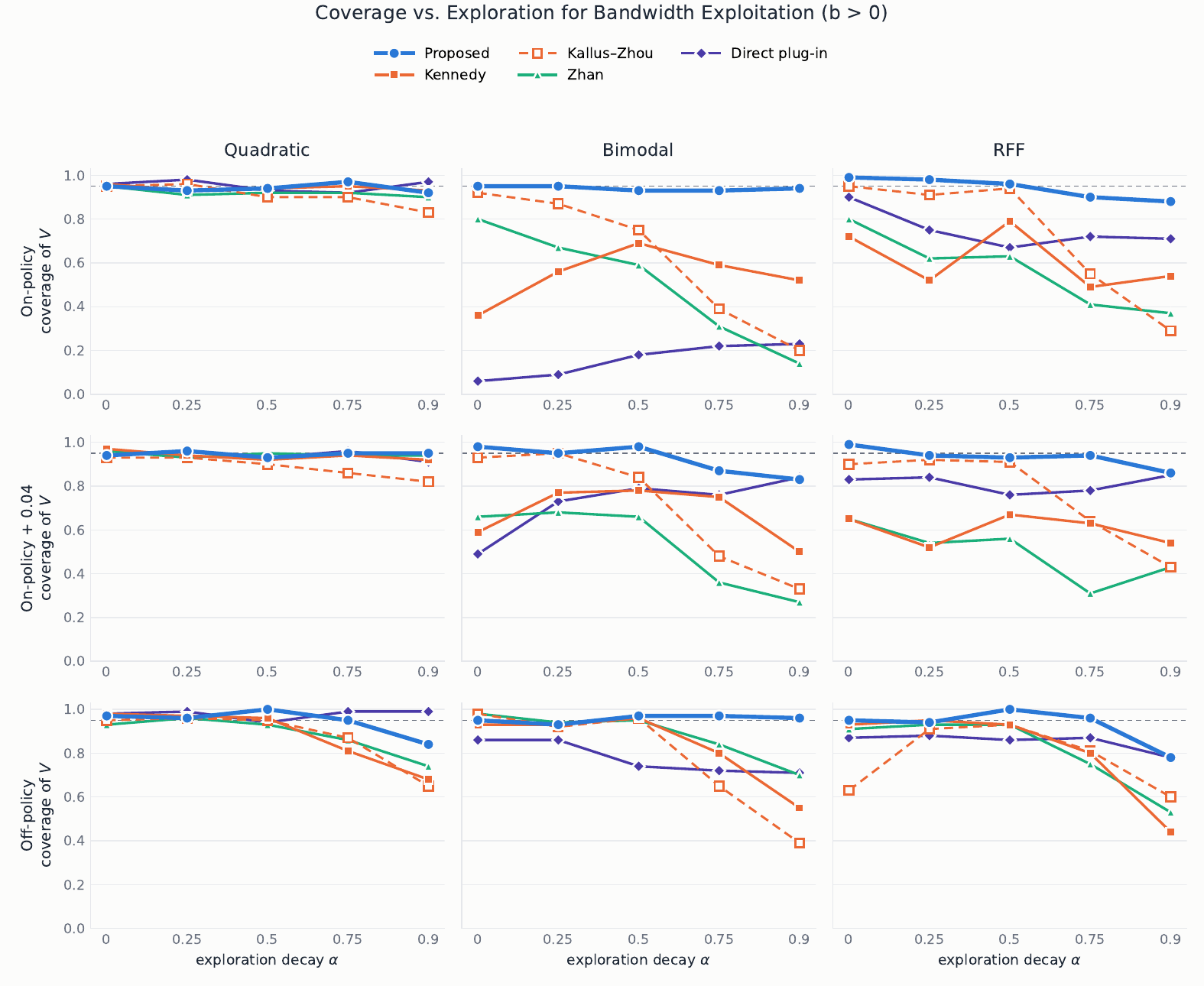}
  \caption{Coverage of nominal 95\% intervals as exploration decays ($\varepsilon_t = t^{-\alpha}$), for all three target policies (rows) and  outcome models (columns). The on-policy $+ 0.04$ target evaluates a policy shifted $0.04$ from the optimum. Bandwidth exploitation with $b = 0.12$; bandwidth by 5-fold forward-chaining cross-validation for all methods; $T = 5000$, 1000 replications per cell. Dashed line: $0.95$.}
  \label{fig:coverage-full}
\end{figure*}

\begin{figure*}[t]
  \centering
  \includegraphics[width=\textwidth]{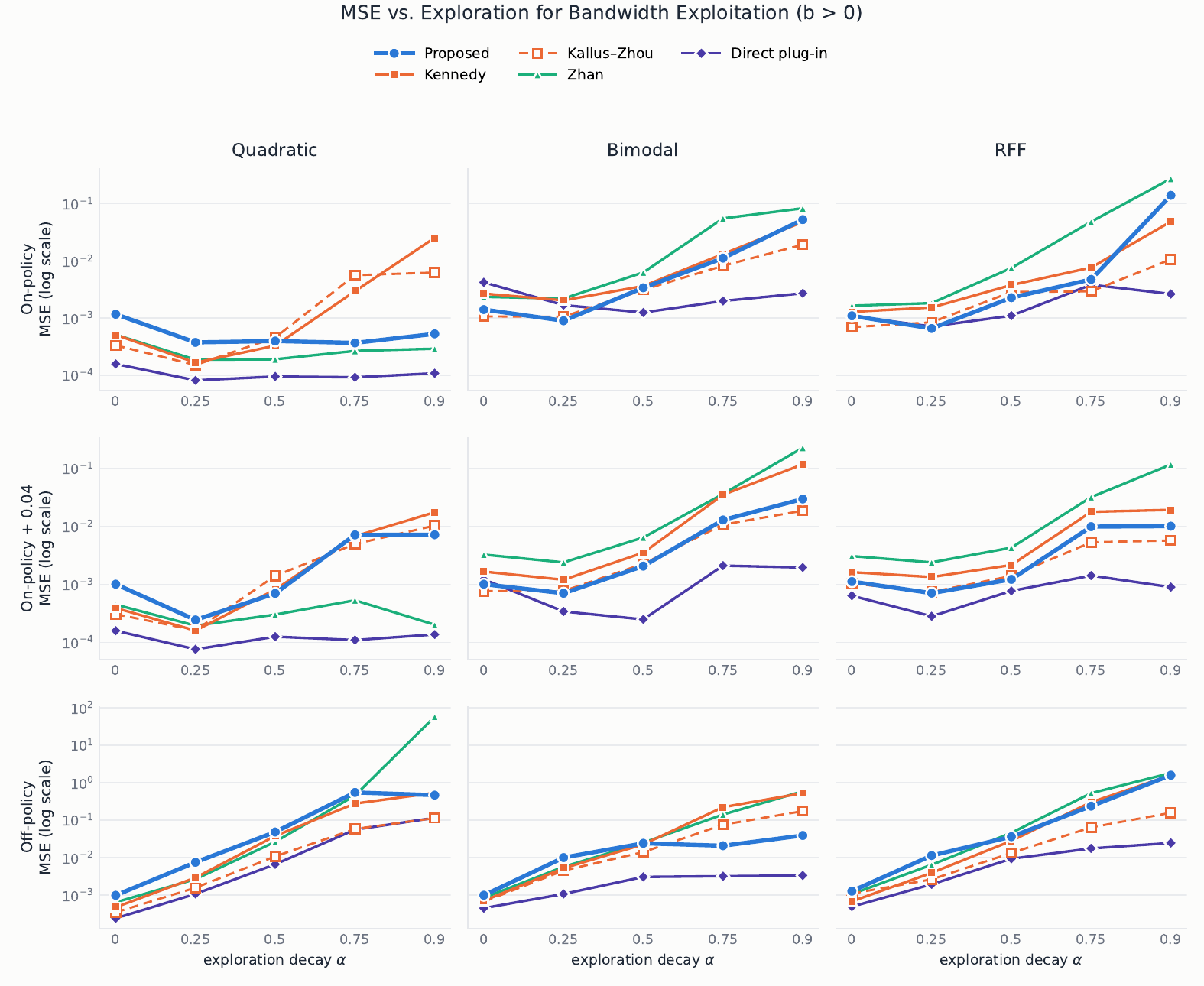}
  \caption{Mean squared error (log scale) of the policy-value estimates; same design as \cref{fig:coverage-full}.}
  \label{fig:mse-full}
\end{figure*}

\subsection{Exploitation Comparison: Bandwidth and Point Dose}\label{sec:atomCmp}

In this section, we compare the exploitation approaches with $b>0$ and $b=0$. We see from \cref{fig:atomCmp}, bandwidth exploitation is critical for inference. Relative to point-mass exploitation, the effective sample size at on-policy targets grows significantly, shrinking interval width by a factor of 3 to 70, and attaining nominal coverage once exploration decays ($\alpha \ge 0.75$). For off-policy targets, in which the two effective sample sizes are similar, the two approaches perform similarly.

\begin{figure*}[ht!]
  \centering
  \begin{subfigure}[t]{0.49\textwidth}
    \centering
    \includegraphics[width=\linewidth]{figs/headline_coverage_onecol.pdf}
    \caption{Bandwith exploitation ($b = 0.12$).}
    \label{fig:cov-jitter}
  \end{subfigure}\hfill
  \begin{subfigure}[t]{0.49\textwidth}
    \centering
    \includegraphics[width=\linewidth]{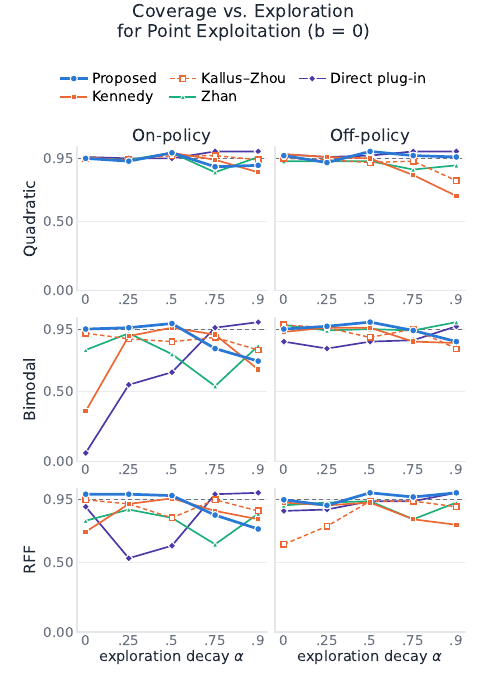}
    \caption{Point-mass exploitation ($b = 0$).}
    \label{fig:cov-atom}
  \end{subfigure}
  \caption{Coverage of nominal 95\% intervals as exploration decays ($\varepsilon_t = t^{-\alpha}$), for the on-policy target (left column of each panel) and an off-policy rule (right column) on three outcome (rows). Bandwidth exploitation increases the effective sample size near the exploited policy, attaining nominal coverage.
  Bandwidth by 5-fold forward-chaining cross-validation; $T = 5000$, 1000 replications per cell. Dashed line: $0.95$.}
  \label{fig:atomCmp}
\end{figure*}

\end{document}